\documentclass[acmsmall,nonacm]{acmart}
\usepackage{etoolbox}
\usepackage{makecell}
\usepackage{tablefootnote}
\usepackage{enumitem}
\usepackage{hyperref}
\usepackage{cleveref}
\usepackage{multicol}
\usepackage{algorithm}
\usepackage{algorithmicx}
\usepackage[noend]{algpseudocode}
\usepackage{graphicx}
\newtheorem{theorem}{Theorem}
\newtheorem{lemma}{Lemma}

\algnewcommand{\NoNumber}[1]{\Statex \hspace*{-\algorithmicindent} #1}
\algdef{SE}[EVENT]{Upon}{EndUpon}[1]{\textbf{upon}\ #1\ \algorithmicdo}{\algorithmicend\ \textbf{event}}%
\algtext*{EndUpon}
\algdef{SE}[EVENT]{AtTime}{EndAtTime}[1]{\textbf{at time}\ #1\ \textbf{do}}{\textbf{end} \textbf{event}}%
\algtext*{EndAtTime}

\newtheorem{rrule}{Rule}
\newtheorem{definition}{Definition}
\setcopyright{none}
\makeatletter
\renewcommand\@formatdoi[1]{\ignorespaces}
\makeatother
\renewcommand{\footnotetextcopyrightpermission}[1]{}

\begin{document}

\title{TetraBFT: Reducing Latency of Unauthenticated, Responsive BFT Consensus}

\author{Qianyu Yu}
\email{qyu100@connect.hkust-gz.edu.cn}
\affiliation{%
  \institution{The Hong Kong University of Science and Technology (Guangzhou)}
  \country{China}}

\author{Giuliano Losa}
\email{giuliano@stellar.org}
\affiliation{%
  \institution{Stellar Development Foundation}
  \country{USA}}

\author{Xuechao Wang}
\email{xuechaowang@hkust-gz.edu.cn}
\affiliation{%
 \institution{The Hong Kong University of Science and Technology (Guangzhou)}
 \country{China}}

\thanks{For correspondence on the paper, please contact Xuechao Wang at xuechaowang@hkust-gz.edu.cn. \newline}

\renewcommand{\shortauthors}{Qianyu Yu, Giuliano Losa and Xuechao Wang}

\begin{abstract}
    This paper presents TetraBFT, a novel unauthenticated Byzantine fault tolerant protocol for solving consensus in partial synchrony, eliminating the need for public key cryptography and ensuring resilience against computationally unbounded adversaries.

    TetraBFT has several compelling features: it necessitates only constant local storage, has optimal communication complexity, satisfies optimistic responsiveness --- allowing the protocol to operate at actual network speeds under ideal conditions --- and can achieve consensus in just 5 message delays, which outperforms all known unauthenticated protocols achieving the other properties listed.
    We validate the correctness of TetraBFT through rigorous security analysis and formal verification.

    Furthermore, we extend TetraBFT into a multi-shot, chained consensus protocol, making a pioneering effort in applying pipelining techniques to unauthenticated protocols. This positions TetraBFT as a practical and deployable solution for blockchain systems aiming for high efficiency.
\end{abstract}



\maketitle

\section{Introduction}

Byzantine fault tolerant (BFT) consensus protocols are the foundation of permissionless blockchain systems like Bitcoin, Ethereum, and Cosmos.
These three systems are each an example of the three types of permissionless blockchains in the hierarchy of Roughgarden and Lewis-Pye~\cite{lewis-pye_permissionless_2023}: fully permissionless (Bitcoin), dynamically available (Ethereum), and quasi-permissionless~(Cosmos).

While quasi-permissionless systems make the strongest assumptions about their participants --- assuming that their list is known and that a large fraction of them (e.g. two thirds) are available and honest --- those assumptions allow using partially-synchronous BFT consensus protocols in the vein of the seminal PBFT protocol~\cite{castro1999practical}.
This is advantageous, because PBFT-style protocols have been honed by at least 3 decades of research and achieve significantly better performance and lower resource consumption than fully-permissionless protocols and, unlike current dynamically-available protocols, they remain safe during asynchrony.

There are two main categories of PBFT-style protocols: protocols relying on authenticated messages, often called authenticated protocols, and protocols relying only on authenticated channels, often called unauthenticated or information-theoretic protocols.
The key difference is that, with authenticated messages, a node $n_1$ can easily prove to a node $n_2$ that a third node $n_3$ sent a given message ($n_1$ just forwards $n_3$'s authenticated message to $n_2$, which can check the authenticity of the message).
However, with only authenticated channels, there is no straightforward way for $n_2$ to verify a claim by $n_1$ that $n_3$ sent a given message --- it could easily be a lie.

Authenticated messages make it fundamentally easier to devise efficient protocols.
For example, in synchronous systems, with signatures we can solve consensus regardless of the number of failures using the Dolev-Strong algorithm~\cite{dolev_authenticated_1983}; without signatures, we must assume that less that a third of the nodes are faulty~\cite{pease_reaching_1980}.
Recent results of Abraham et al.\cite{abraham_good-case_2021,abraham_good-case_2022} also show that Byzantine reliable broadcast is solvable with a good-case latency of 2 message delays in the authenticated model with $3f+1\leq n$ (where $n$ is the number of nodes and $f$ the number of Byzantine nodes among them), but it takes at least 3 message delays in the unauthenticated model when $3f+1\leq n<4f$.
Perhaps for this reason, authenticated protocols have attracted considerable research attention and have seen wide adoption by quasi-permissionless blockchain systems (Cosmos, Sui, Aptos, Solana, etc.).
Example protocols include Tendermint~\cite{buchman2016tendermint}, Hotstuff~\cite{yin2019hotstuff} and its variants~\cite{jalalzai2023fast,malkhi2023hotstuff}, and Jolteon~\cite{gelashvili2022jolteon}, etc.

In contrast, unauthenticated BFT consensus protocols~\cite{abraham2020information,lokhava2019fast,attiya2023multi} have been much less studied.
However, they offer compelling advantages: %
\begin{enumerate}
    \item Unauthenticated BFT protocols avoid computationally expensive cryptographic operations, like public-key cryptography, and are therefore advantageous in resource-constrained environments.
        This includes protocols which must be executed by smart contracts, e.g.\ trustless cross-chain synchronization protocols like TrustBoost~\cite{sheng2023trustboost} and the Interchain timestamping protocol~\cite{tas2023interchain}.
        In those settings, public-key cryptography would be prohibitively expensive.
    \item Instead of mandating global trust assumptions (such as 1/2 of the mining power being honest in Bitcoin, or 1/2 of the weighted stakers being honest in Ethereum), some blockchains allow participants to make unilateral, heterogeneous trust assumptions.
        For example, one participant might assume 1/2 of a set $A$ of participants is honest, while another participant assumes 2/3 of another set $B\neq A$ is honest.
        Examples include Ripple\cite{amores2020security} and Stellar's Federated Byzantine Agreement model (FBA)~\cite{lokhava2019fast}.
        In this setting, authenticated messages are of limited use because forwarding the set of messages that caused one node $n_1$ to take some action will not necessarily convince another node $n_2$ to take the same action if $n_2$'s trust assumptions are different from $n_1$.
        For example, well-known techniques based on disseminating quorum certificates do not work anymore in this setting.
        This rules out every authenticated protocol that we know of.
        However, with a few tweaks, unauthenticated protocols work in this heterogeneous setting.
    \item Finally, rehashing Abraham and Stern~\cite{abraham2020information}, by avoiding message authentication, unauthenticated protocol implementations have a smaller attack surface and rely on minimal cryptographic assumptions, making them more ``future-proof''.
\end{enumerate}

In this paper, we are interested in unauthenticated, partially-synchronous BFT consensus protocols that have optimal resilience (i.e.\ requiring only $n\geq 3f+1$ nodes), that are optimistically-responsive~\cite{pass_thunderella_2018}, that use only constant local storage space, and that have an optimal communication complexity of $O(n^2)$ bits~\cite{dolevBoundsInformationExchange1985}, where each node sends and receives a linear number of bits.
This combination of features is almost a requirement in practice:
\begin{itemize}
    \item Optimal resilience means we maximize the fault-tolerance budget given a fixed system size.
        This is crucial in the extremely adversarial environments faced by permissionless systems.
    \item Once the network becomes synchronous, optimistically responsive protocols make progress as fast as messages are received; instead, their non-responsive counterparts~\cite{li2023quorum,Non-IT-HS} must wait for fixed-duration timeouts to recover from asynchrony.
        In practice, non-responsiveness can cause large performance hiccups, leading to backlogs of work that are hard to clear out.
    \item In a blockchain setting, where the system is expected to never stop and potentially reach very large scale (e.g. Ethereum currently has hundreds of thousands of validators), we cannot tolerate storage requirements that keep growing forever or a communication complexity higher than linear per node.
\end{itemize}

Before this work, there were only 3 known optimally-resilient, optimistically-responsive, partially-synchronous unauthenticated BFT consensus protocols: two versions of PBFT~\cite{castro_thesis,castro2002practical} (a version requiring unbounded storage and a constant-storage version) and Information-Theoretic HotStuff (IT-HS)~\cite{abraham2020information}.
While the unbounded-storage, unauthenticated PBFT version obviously fails our bounded-storage requirement, its constant-storage counterpart fails the linear-communication requirement: it uses a complex view-change protocol in which each nodes sends a worst-case quadratic number of bits, for a worst-case cubic total complexity.
Only IT-HS achieves all our desiderata (using both constant storage and linear per-node worst-case communication).
However, IT-HS has a whooping good-case latency of 6 message delays.
Roughly speaking, good-case latency is the latency of the protocol when the network is synchronous and the system is not under attack.
It is important to minimize it because this will be the latency in the common case in practice.

It is not known whether IT-HS's good-case latency of 6 message delays can be improved without incurring non-constant space usage (either in messages or local persistent state).
Thus, the first question that we ask in this paper is the following:
\begin{quote}
    \textit{What is the minimum good-case latency achievable in partially-synchronous, optimally-resilient, constant-space, quadratic-communication, optimistically-responsive unauthenticated BFT consensus?}
\end{quote}

We make significant progress towards an answer by presenting TetraBFT, an unauthenticated, partially synchronous BFT consensus algorithm that has optimal resilience ($n\geq 3f+1$), is optimistically-responsive, uses constant-size storage, has worst-case quadrati communication, and has a good-case latency of 5 message delays.
This shows that IT-HS's 6 message delays is not optimal, but we do not know whether 5 is the lower bound.

We provide detailed proofs establishing both the safety and liveness of TetraBFT, but we also formally specify TetraBFT using the TLA+ language~\cite{lamport_specifying_2002} and we mechanically verify its safety property using the Apalache model-checker~\cite{konnov_tla_2019}.
We are able to exhaustively check that TetraBFT is safe in all possible executions with 4 nodes, among which one is Byzantine, 3 different proposed values, and a maximum of 5 views.
This gives us extremely high confidence in the safety of TetraBFT.

While unauthenticated BFT protocols satisfying the desiderata above are a first step towards practical protocols, the next step in practice is to boost throughput using pipelining.
However, the very brief sketch in~\cite{abraham2020information} notwithstanding, there are no known pipelined unauthenticated BFT consensus protocols.
The next question we ask is therefore:
\begin{quote}
    \textit{Can unauthenticated BFT consensus protocols be made practical using pipelining?}
\end{quote}

We answer positively by extending TetraBFT for pipelined multi-shot consensus (also known as state machine replication (SMR)), enabling nodes to efficiently reach consensus on a sequence of blocks and thus construct a blockchain.
Using pipelining, TetraBFT is able to commit one new block every message delay in the good case, and thus, in theory, it achieves a maximal throughput of 5 times the throughput that would be achieved by simply repeating instances of single-shot TetraBFT.
Pipelined TetraBFT is also conceptually simple, using only 2 message types in the good case (proposals and votes), and uses its view-change protocol only to recover from a malicious leader or from asynchrony.
This is a significant improvement over the pipelined version of IT-HS (briefly sketched in Section 3.2 of~\cite{abraham2020information}), which mandates the transmission of \texttt{suggest}/\texttt{proof} messages (used during view change in the single-shot IT-HS) alongside \texttt{vote} messages regardless of the scenario.
To the best of our knowledge, this work is the first to offer a comprehensive description and analysis of pipelining in the unauthenticated setting.

Single-shot TetraBFT, described in~\Cref{sec:single-shot}, follows the classic blueprint for partially-synchronous consensus first proposed by Dwork, Lynch, and Stockmeyer~\cite{dwork1988consensus}.
The protocol executes a sequence of views, where each view has a unique leader that proposes a value, followed by a voting phase after which either a decision is made or nodes timeout and start a view-change protocol.
In the view-change protocol, nodes gather information about what happened in previous views and determine which values are safe --- meaning values that cannot possibly contradict any decision that was or will ever be made in a previous view --- to propose and vote for in the new view.

Our main contributions are summarized below.
\begin{enumerate}
    \item We propose TetraBFT, a novel partial synchronous unauthenticated BFT consensus algorithm with bounded persistent storage, optimistic responsiveness, and reduced latency, improving efficiency by shortening consensus phases.
    \item We conduct a comprehensive security analysis and a formal verification of TetraBFT, affirming the correctness of the protocol and its established security properties.
    \item We extend TetraBFT into a multi-shot, pipelined consensus algorithm, making the first detailed exploration of pipelining within the unauthenticated context. %
\end{enumerate}

\subsection{Protocol Overview}

We consider a message-passing system consisting of $n$ nodes running TetraBFT.
Among the $n$ nodes, $f$ exhibiting Byzantine faults are considered malicious, while the remaining $n-f$ nodes are well-behaved.
We assume $3f < n$.
Well-behaved nodes follow the protocol faithfully while malicious nodes may arbitrarily deviate from the protocol.
We define any group of $n-f$ or more nodes as a \textit{quorum}, and any group of $f+1$ or more nodes as a \textit{blocking set}.

TetraBFT executes a sequence of views, each of which is pre-assigned a unique leader (e.g.\ chosen round-robin).
Each view consists of 7 phases leading to a possible decision: \texttt{suggest}/\texttt{proof}, \texttt{proposal}, \texttt{vote-1}, \texttt{vote-2}, \texttt{vote-3}, \texttt{vote-4} and \texttt{view-change}.
Initially, at view 0, \texttt{suggest}
/\texttt{proof} messages are not required as all the values are determined safe.
In the good case, no \texttt{view-change} message is sent, allowing a decision to be reached in 5 phases.

The protocol is segmented into four key components:

1. \textbf{Proposal.}
A well-behaved leader determines the safety of a value based on the \texttt{suggest} messages from a quorum of nodes before proposing.
Only values determined safe are proposed.
A \texttt{suggest} message is composed of the historical information about \texttt{vote-2} and \texttt{vote-3} messages.

2. \textbf{Voting.}
Mirroring the leader's process, well-behaved nodes determine the safety of a proposed value through a quorum of \texttt{proof} messages, which carry historical information about \texttt{vote-1} and \texttt{vote-4} messages.
A node casts \texttt{vote-1} only for values determined safe.
A node votes through the voting sequence—\texttt{vote-1}, \texttt{vote-2}, \texttt{vote-3}, to \texttt{vote-4}—sequentially, sending each subsequent vote only after receiving a quorum of messages for the preceding vote type.
The name TetraBFT comes from the fact that there are 4 voting phases in the protocol.

3. \textbf{Deciding.}
A node decides a value $val$ upon receiving a quorum of \texttt{vote-4} messages for $val$.

4. \textbf{View change.}
A node sends a \texttt{view-change} message for the next view when its current view does not produce a decision  by a fixed time after the view started (nodes use timers for this).
A node also issues a \texttt{view-change} message for a new view $v$ upon receiving $f+1$ messages from that view, provided it hasn't already sent a message for $v$ or a higher view.
A transition to the new view occurs once $n-f$ \texttt{view-change} messages are received. Nodes keep checking \texttt{view-change} messages throughout all views.

In a nutshell, TetraBFT guarantees safety by ensuring the well-behaved nodes only ever vote for safe values, i.e.\ values that cannot possibly contradict a decision in a previous view.
Liveness is more subtle, and the key is that the leader determines whether a value $val$ is safe at a view when the \texttt{vote-2} messages of a blocking set show that $val$ is safe, and follower nodes determine that $val$ is safe when the \texttt{vote-1} messages of a blocking set show that $val$ is safe.
Thus, any value determined safe by a leader is also determined safe by all other well-behaved nodes.

\subsection{Related Work}

\Cref{table:comparison} compares TetraBFT to other partially-synchronous, unauthenticated BFT protocols in terms of responsiveness (where ``responsive'' in the table means optimistically responsive~\cite{pass_thunderella_2018}), good-case latency, latency with view-change, and size of storage and messages.
Good-case latency means the latency of the protocol in message delays when the network is synchronous from the start and the leader of the first view is well-behaved. Latency with view-change is the latency of a view, in message delays, starting with a view-change.

We justify latencies as follows.
IT-HS (blog version)~\cite{Non-IT-HS} has a latency of 4 phases in the good case: propose, echo, accept, lock, and a suggest message is sent in the view-change scenario.
The Byzantine consensus protocol in Li et al.~\cite{li2023quorum} has a latency of 6 in both cases, attributable to the employment of two instances of Byzantine reliable broadcast, each consisting of three phases.
IT-HS~\cite{abraham2020information} experiences a latency involving 6 phases in the good case: propose, echo, key-1, key-2, key-3, and lock.
Additionally, proof/suggest, request, and abort messages become necessary in the view-change case.
PBFT~\cite{castro_thesis,castro2002practical} demonstrates a latency of 3 phases in the good case: pre-prepare, prepare, commit.
The view-change case necessitates an extra four messages: request, view-change, view-change-ack, and new-view.

\begin{table}[h]
    \caption{Characteristics of the partially-synchronous, unauthenticated BFT consensus protocols known to the authors, as well as two protocols (SCP and the protocol of Li et al.) for heterogeneous-trust systems. ``Responsive'' means optimistically responsive. Latencies are expressed in message delays.}
    \vspace{-3mm}
    \label{table:comparison}
    \begin{tabular}{c c c c c}
        \hline
    & \makecell{Responsiveness} & \makecell{Good-case\\ latency} & \makecell{Latency with\\ view-change } & \makecell{Storage/\\Communicated bits}\\ [0.5ex]
    \hline\hline
        IT-HS (blog version)~\cite{Non-IT-HS} & non-responsive &4& 5&$O(1)$/$O(n^2)$\\
        \hline
        IT-HS~\cite{abraham2020information} & responsive & 6 & 9 & $O(1)$/$O(n^2)$\\
        \hline
        PBFT (bounded)~\cite{castro_thesis} & responsive & 3 & 7 & $O(1)$/$O(n^3)$ \\
        \hline
        PBFT (unbounded)~\cite{castro_thesis,castro2002practical} & responsive & 3 & 7 & unbounded/unbounded\\
        \hline
        SCP~\cite{lokhava2019fast} & not applicable\footnotemark & 6 & 4\footnotemark & $O(1)$/$O(n^2)$\\
        \hline
        Li et al.~\cite{li2023quorum} & non-responsive & 6& 6 & unbounded/unbounded\\
        \hline
        \textbf{TetraBFT} & responsive & 5 & 7 & $O(1)$/$O(n^2)$ \\
        \hline
    \end{tabular}
\end{table}

\addtocounter{footnote}{-1}
\footnotetext{SCP does not guarantee termination unless Byzantine nodes are all eventually evicted.}
\addtocounter{footnote}{1}
\footnotetext{In SCP, views after the first take fewer phases but have weaker liveness guarantees than the first view.}

As we can see in \Cref{table:comparison}, unresponsive protocols have better latency.
However, non-responsiveness is problematic in practice.
Optimistically responsive means that, once the network becomes synchronous with actual delay $\delta$, all well-behaved parties decide in time proportional to $\delta$ (at most $7\delta$ in TetraBFT) instead of proportional to the worst-case latency $\Delta$ in the non-responsive case; this is typically due to the leader having to wait for a fixed duration at the beginning of a new view in order to collect enough information in order to make a proposal that will be accepted by the other nodes.
In practice, one usually makes conservative assumptions about message delay, such that it is likely for $\delta$ to be much smaller than $\Delta$.
In this case, an unresponsive protocols will take significantly longer to change views.
This is a practical problem because, in multi-shot consensus, a long view-change will cause an accumulating backlog of work that may be hard to recover from.

PBFT, IT-HS, and TetraBFT are optimistically responsive.
PBFT achieves better latency, but even its constant-storage version~\cite{castro_thesis} still sends $O(n^3)$ messages in the worst case.
This is not practical in large systems.
For example, Ethereum has hundreds of thousands of (logical) validators as of early 2024.

Like TetraBFT, IT-HS is optimistically responsive, uses constant storage, and has an optimal communication complexity of $O(n^2)$ bits per view.
The main advantage of TetraBFT over IT-HS is that TetraBFT reduces good-case latency by one message-delay.
TetraBFT achieves this feat by ensuring that, upon view change, well-behaved nodes collect enough information to never send a message containing an unsafe value. %
Instead, IT-HS relies on locks for safety, but locks are in some sense imperfect: even though an unsafe value cannot make it to the ``key-1'' phase of IT-HS because a quorum will have a lock for a different value, some well-behaved nodes may still echo unsafe values because they are not locked.
Thus, even if $f+1$ nodes echo the same value $val$, it does not prove that $val$ is safe, and this property is only achieved at the ``key-1'' phase of IT-HS.
In contrast, in TetraBFT, well-behaved nodes never send a message containing an unsafe value.
Thus we do not need the equivalent of IT-HS's echo phase, and that is how TetraBFT achieves 5 phases instead of 6.

To the best of our knowledge, the brief sketch in~\cite{abraham2020information} notwithstanding, this work is the first to offer a comprehensive description and analysis of pipelining in the unauthenticated setting.
Abraham and Stern~\cite{abraham2020information} briefly discuss pipelining IT-HS, but do not give much detail.
Moreover, their suggested protocol seems to mandate the transmission of \texttt{suggest}/\texttt{proof} messages (used during view change in the single-shot IT-HS) alongside \texttt{vote} messages regardless of the scenario, while Multi-shot TetraBFT sends only proposals and votes in the good case.

The Stellar Consensus Protocol (SCP)~\cite{lokhava2019fast} and the protocols of Li et al.~\cite{li2023quorum} are not strictly speaking unauthenticated  protocols.
While they are partially-synchronous BFT consensus protocols, they are presented in two heterogeneous trust models: the federated Byzantine agreement (FBA) model for SCP, and the heterogeneous quorum system model for Li et al.
In these models, nodes are allowed to make their own failure assumptions, resulting in different nodes having different sets of quorums.
Variants of heterogeneous models include asymmetric quorum systems~\cite{cachin_asymmetric_2021}, subjective quorum systems~\cite{garcia-perez_federated_2018}, personal Byzantine quorum systems~\cite{losa_stellar_2019}, and permissionless quorum systems~\cite{cachin_quorum_2023}.
In these models, because nodes may not agree on what is a quorum, authenticated messages and quorum certificates --- as used in one form or another in all partially-synchronous, authenticated BFT consensus protocols known to the authors --- do not work.
Thus, for the main aspects of BFT consensus protocols, those heterogeneous models look unauthenticated, and SCP and the protocol of Li et al.\ indeed are easily transferable to the unauthenticated setting.
An implementation of SCP, called stellar-core, has been in use in the Stellar network since 2015.

An interesting observation is that this also works the other way, from the unauthenticated setting to the heterogeneous setting: TetraBFT could be adapted to work in an heterogeneous setting like the FBA model.
The main difficulty in open heterogeneous settings like FBA is assigning a unique leader to each view.
This is because, without global agreement on a list of participants, we cannot use a round-robin strategy.
Instead, SCP uses a synchronous sub-protocol, called the nomination protocol~\cite{lokhava2019fast}, whose principles could be applied to TetraBFT to obtain simulate a unique leader.

\section{Model and Preliminary Definitions}
We assume the classic setting with $n>3f$ nodes among which $f$ are Byzantines.
Each node has a local timer (used for timeouts) ticking at the same rate and we assume that local computation is instantaneous. 
However, the network is only partially synchronous~\cite{dwork1988consensus}.
This means that the network is initially asynchronous but becomes synchronous after an unknown global stabilization time, noted GST.
Before GST, there is no guarantee of message delivery and messages sent before GST may be permanently lost. %
Note that with constant storage, we must allow the loss of messages during asynchrony, as preventing this would necessitate unbounded buffers.
However, every message sent after GST is guaranteed to be delivered within a known bound $\Delta$.

We are interested in protocols solving first the consensus problem, and then its multi-shot variant total-order broadcast (TOB, also called atomic broadcast).
In the problem of consensus, each node starts with an initial input value, with the goal being for all nodes to agree on a single output value.
A distributed algorithm solving consensus should guarantee the usual properties:

\begin{definition}[Consensus]\label{def:concensus}
\leavevmode
\begin{itemize}
\item \textbf{Termination.} Every well-behaved node eventually decides a value.
\item \textbf{Agreement.} No two well-behaved nodes decide different values.
\item \textbf{Validity.} If all nodes are well-behaved and they all have the same input value $val$, then every well-behaved node that decides a value decides the value $val$.
\end{itemize}
\end{definition}

In contrast to consensus, multi-shot consensus enables nodes to reach consensus on an unlimited series of values $val1, val2, val3,\ldots$
Each value is assigned with a \textit{slot} number, indicating its position in the series.
In the context of a blockchain, these values, essentially data blocks containing \textit{transactions}, are linked sequentially via hash pointers, collectively forming a \textit{chain}.
When a node outputs a block or an entire chain, we say that the block or the chain is \textit{finalized} by the node.
Multi-shot consensus satisfies consistency and liveness, as defined in~\cite{chan2020streamlet}:

\begin{definition}[Multi-shot Consensus]
\leavevmode
\begin{itemize}
\item \textbf{Consistency}. If two chains are finalized by two well-behaved nodes, then one chain must be a prefix of, or equal to, the other.
\item \textbf{Liveness}. If some well-behaved node receives a transaction $\mathsf{txn}$, $\mathsf{txn}$ will eventually be included in all well-behaved nodes’ finalized chains.
\end{itemize}
\end{definition}

\section{Basic TetraBFT}
\label{sec:single-shot}
In this section, we present the Basic TetraBFT protocol that solves the problem of consensus as defined in Definition~\ref{def:concensus}.
We first delve into the message types used in the protocol, detailing the structure and the purpose of each message type. Following this, we present a comprehensive operational description of the protocol, focusing on the evolution of a view. Finally, we provide two helper algorithms designed to assist nodes in efficiently determining safe values.

\subsection{Messages}
We outline the message types that a node can send in Basic TetraBFT. We denote a view by $v$ and a value by $val$. A leader node sends \texttt{proposal} messages and non-leader nodes can send 4 types of \texttt{vote} messages, \texttt{suggest/proof} messages and \texttt{view-change} messages. The \texttt{suggest/proof} messages, incorporating historical records of previously sent  \texttt{vote} messages, facilitate leaders and nodes in determining safe values. Besides, \texttt{view-change} messages enable nodes to transition to a new view. In a more formal structure:

\begin{itemize}
\item
Only sent by the leader:
\begin{itemize}
\item \texttt{proposal} message: formatted as $\langle\texttt{proposal}, v, val\rangle$.
\end{itemize}
\item
Sent by all nodes:
\begin{itemize}
  \item \texttt{vote-i} message: $\langle\texttt{vote-i}, v, val\rangle$, where $\texttt{i = 1,2,3,4}$.
  \item \texttt{suggest} message: $\langle\texttt{suggest}, \langle \texttt{vote-2}, v, val \rangle, \langle \texttt{prev-vote-2}, v, val' \rangle, \langle \texttt{vote-3}, v, val \rangle \rangle$, where
    \begin{itemize}
      \item highest \texttt{vote-2} message $\langle \texttt{vote-2}, v, val \rangle$: the highest \texttt{vote-2} message in the view number that the node has sent;
      \item second-highest \texttt{vote-2} message $\langle \texttt{prev-vote-2}, v, val' \rangle$: the highest \texttt{vote-2} message in the view number that the node has sent for a different value from the highest \texttt{vote-2} message;
      \item highest \texttt{vote-3} message $\langle \texttt{vote-3}, v, val \rangle$: the highest \texttt{vote-3} message in the view number that the node has sent.
    \end{itemize}
  \item \texttt{proof} message: $\langle\texttt{proof}, \langle \texttt{vote-1}, v, val \rangle, \langle \texttt{prev-vote-1}, v,val' \rangle, \langle \texttt{vote-4}, v, val \rangle \rangle$, which follows the same struture as the \texttt{suggest} message, but uses \texttt{vote-4} instead of \texttt{vote-3} and \texttt{vote-1} instead of \texttt{vote-2}.
   \item \texttt{view-change} message: $\langle \texttt{view-change}, v \rangle$.
\end{itemize}
\end{itemize}

Throughout the views, a node needs only to store the highest \texttt{vote-1} and \texttt{vote-2}, \texttt{vote-3} and \texttt{vote-4} messages it sent, along with the second highest \texttt{vote-1} and \texttt{vote-2} messages that carry a different value from their respective highest messages. Thus, similar to IT-HS, TetraBFT requires only a constant amount of persistent storage.

\subsection{Evolution of a View}
\label{sec:evolution-view}
Our protocol operates on a view-based manner, with each view having a unique and pre-determined leader. Every node begins at view 0, equipped with an initial value. In any view, should the leader determine that arbitrary values (including its initial value) are safe in step~\ref{step:2}, it will propose its initial value by default.
A view $v$ proceeds as follows:
\begin{enumerate}
  \item Upon starting the view, each node sets its timer to a timeout of $9\Delta$, ensuring sufficient time for deciding when the leader of current view is well-behaved.
  If $v=0$, the node proceeds directly to step~\ref{step:2}; otherwise, if $v>0$, the node undertakes the following steps:
    \begin{enumerate}
      \item it broadcasts a \texttt{proof} message for the current view and
      \item it sends a \texttt{suggest} message to the leader of the current view.
    \end{enumerate}
  \item \label{step:2} When the leader has determined that a value $val$ is safe to propose
    in the current view according to \Cref{rule:picking-safe-proposal}, it broadcasts a
    \texttt{proposal} message for the current view and for $val$. {A well-behaved leader broadcasts only one \texttt{proposal} message in a view.}
  \item \label{step:3} When a node determines that the leader's proposal is safe in the
    current view according to \Cref{rule:checking-safe-proposal}, it broadcasts a \texttt{vote-1}
    message for the current view and the leader's proposal. {A well-behaved node broadcasts only one \texttt{vote-1} message in a view.}
    \item A node that receives a quorum of \texttt{vote-1} messages for the
      current view and for the same value $val$ sends a \texttt{vote-2}
      message for the current view and for $val$.
    \item A node that receives a quorum of \texttt{vote-2} messages for the
      current view and for the same value $val$ sends a \texttt{vote-3}
      message for the current view and for $val$.
    \item A node that receives a quorum of \texttt{vote-3} messages for the
      current view and for the same value $val$ sends a \texttt{vote-4}
      message for the current view and for $val$.
    \item A node that receives a quorum of \texttt{vote-4} messages for the
      current view and for the same value $val$ decides $val$. 
\end{enumerate}

Upon observing the timer expiration, a node broadcasts a \texttt{view-change} message for the next view. On receiving $f+1$ \texttt{view-change} messages for a view $v'$, a node sends a \texttt{view-change} message for view $v'$ if it has not sent a \texttt{view-change} message for view $v'$ or any higher view.  On receiving $n-f$ \texttt{view-change} messages for a view, a node changes to the view. Nodes keep checking \texttt{view-change} messages regardless of the current view.

We justify the timeout value of $9\Delta$, assuming the network is synchronous with a maximum delay of $\Delta$ (i.e., after GST). Let us consider a well-behaved node that receives $n-f$ \texttt{view-change} messages at time $t$, including at least $f+1$ from well-behaved nodes. Due to the network delay, other well-behaved nodes might receive these $f+1$ \texttt{view-change} messages at time $t+\Delta$ and subsequently echo a \texttt{view-change} message. Hence, all well-behaved nodes should receive $n-f$ \texttt{view-change} messages by time $t+2\Delta$. This indicates a maximum difference of $2\Delta$ in the view change time across well-behaved nodes. An additional $6\Delta$ is necessary for processing \texttt{suggest/proof} messages, a \texttt{proposal} message, and four \texttt{vote} messages. Thus, summing these intervals, we opt for a timeout of $9\Delta$ to slightly overshoot the cumulative $8\Delta$, adding a safety margin.

The following are the rules applied in steps~\ref{step:2} and~\ref{step:3}. \Cref{rule:picking-safe-proposal} applies when a well-behaved leader determines a value safe based on \texttt{suggest} messages from a quorum of nodes. 
If $v\neq 0$, the quorum must meet specific criteria, including the existence of a blocking set within the quorum all claiming in their \texttt{suggest} messages that the value is safe according to \Cref{rule:claims_safe_leader}.
Note that ``claim'' is distinguished from ``determine'' used previously. \Cref{rule:checking-safe-proposal} is invoked when a well-behaved node determines a leader's proposal safe, mirroring \Cref{rule:picking-safe-proposal} but with notable distinctions: it relies on \texttt{proof} messages instead of \texttt{suggest} message, employs \texttt{vote-4} in place of \texttt{vote-3}, and \texttt{vote-1} instead of \texttt{vote-2}. Additionally, \Cref{rule:checking-safe-proposal} \Cref{2biii} incorporates an extra condition involving two blocking sets of nodes, which corresponds to a special scenario where any value can be determined safe. In Section~\ref{sec:proof}, we will proof safety and liveness properties based on these rules. Specifically, we will justify the rationale behind \Cref{rule:picking-safe-proposal} and demonstrate how \Cref{rule:checking-safe-proposal} can be logically derived from \Cref{rule:picking-safe-proposal}.

\begin{rrule}
\label{rule:picking-safe-proposal}
All values are safe in view 0. If $v\neq 0$, a leader determines that the value $val$ is safe to propose in view
  $v$ when the following holds:
  \begin{enumerate}
    \item A quorum $q$ has sent \textup{\texttt{suggest}} messages in view $v$, and
    \item According to what is reported in \textup{\texttt{suggest}} messages, either
      \begin{enumerate}
        \item no member of $q$ sent any \textup{\texttt{vote-3}} before view
          $v$, or
        \label{case:no_vote_leader}
        \item there is a view $v'< v$ such that
          \label{case:highest_vote_leader}
          \begin{enumerate}
            \item \label{case:case 1} no member of $q$ sent any \textup{\texttt{vote-3}} messages for a
              view strictly higher than $v'$, and
            \item \label{case:case 2} any member of $q$ that sent a \textup{\texttt{vote-3}} message in view $v'$ did so with value $val$, and
            \item \label{case:case 3} there is a blocking set $b$ (e.g. $f +1$
              nodes) that all claim in their \textup{\texttt{suggest}} messages that $val$ is safe at $v'$ (see \Cref{rule:claims_safe_leader}).
          \end{enumerate}
      \end{enumerate}
  \end{enumerate}
\end{rrule}

\begin{rrule}
  \label{rule:claims_safe_leader}
We say that a node claims that $val$ is safe in $v'$ in a
\textup{\texttt{suggest}} message when either
\begin{enumerate}
  \item $v'$ is 0,
  \label{view 0}
  \item the node's highest \textup{\texttt{vote-2}} message, as reported in the
    \textup{\texttt{suggest}} message, was sent at view $v''\geq v'$
    and for value $val$, or
  \label{node claims safe in suggest message highest view}
  \item the second highest view for which the node sent a \textup{\texttt{vote-2}} message, as reported in the \textup{\texttt{suggest}} message, is a view $v''\geq
    v'$.
  \label{node claims safe in suggest message second highest view}
\end{enumerate}
\end{rrule}

\begin{rrule}
\label{rule:checking-safe-proposal}

A node that receives a proposal from the leader of the current view determines
that the value $val$ is safe to propose in view
$v$ when:
\begin{enumerate}
  \item A quorum $q$ has sent \textup{\texttt{proof}} messages in view $v$, and
  \item According to what is reported in \textup{\texttt{proof}} messages, either
    \begin{enumerate}
      \item no member of $q$ sent any \textup{\texttt{vote-4}} before view
        $v$, or
        \label{case:no_votes}
      \item there is a view $v'< v$ such that
        \label{case:highest_vote}
        \begin{enumerate}
          \item\label{2bi} no member of $q$ sent any \textup{\texttt{vote-4}} messages for a view strictly higher than $v'$, and
          \item\label{2bii} any member of $q$ that sent a \textup{\texttt{vote-4}} message in
            view $v'$ did so with value $val$, and
          \item\label{2biii} 
          \begin{enumerate}[align=left]
            \item \label{case:one block set} there is a blocking set $b$ (e.g. $f +1$ nodes) that all claim in their \textup{\texttt{proof}} messages that
            $val$ is safe at $v'$ (as described in \Cref{rule:claims_safe}), {or
            \item \label{case:two block sets} there is a value $\widetilde{val}$, a view $\tilde{v}$, and blocking set $\tilde{b}$ that all claim in their \textup{\texttt{proof}} messages that
            $\widetilde{val}$ is safe at $\tilde{v}$ where $v'\le \tilde{v}< v$, and} 
            
            {there is a value $\widetilde{val}'$, a view $\tilde{v}'$, and blocking set $\tilde{b}'$ that all claim in their \textup{\texttt{proof}} messages that
            $\widetilde{val}'\neq \widetilde{val}$ is safe at $\tilde{v}'$ where $\tilde{v} < \tilde{v}'< v$.}
           \end{enumerate}
        \end{enumerate}
    \end{enumerate}
\end{enumerate}
\end{rrule}

\begin{rrule}
  \label{rule:claims_safe}
We say that a node claims that $val$ is safe in $v'$ in a
\textup{\texttt{proof}} message when either
\begin{enumerate}
  \item $v'$ is 0,
    \label{case:r'=0}
  \item the node's highest \textup{\texttt{vote-1}} message, as reported in the
    \textup{\texttt{proof}} message, was sent at view $v''\geq v'$
    and for value $val$, or
  \label{case:highest vote-1}
  \item the second highest view for which the node sent a \textup{\texttt{vote-1}} message, as reported in the \textup{\texttt{proof}} message, is a view $v''\geq v'$.
  \label{case:second highest vote-1}
\end{enumerate}
\end{rrule}

\subsection{Helper Algorithms}
In this section, we introduce two helper algorithms that assist nodes in efficiently determining safe values according to \Cref{rule:picking-safe-proposal} and \Cref{rule:checking-safe-proposal} respectively. In \Cref{sec:Liveness Argument}, we will prove the existence of a safe proposal according to \Cref{rule:picking-safe-proposal}: upon receiving \texttt{suggest} messages from a quorum that includes all well-behaved nodes, a well-behaved leader is able to determine a value to be safe (see \Cref{claim:safety existence}). However, it is important to note that a leader does not know whether the \texttt{suggest} messages it receives are from well-behaved nodes or not. Therefore, the leader should carefully select a quorum from at most $n$ \texttt{suggest} messages to determine a safe value. This process necessitates an efficient algorithm, which we present as \Cref{alg:leader finds v} for \Cref{rule:picking-safe-proposal}. Similarly, for \Cref{rule:checking-safe-proposal}, we introduce \Cref{alg:node determine value}, which allows nodes to efficiently identify a quorum from at most $n$ \texttt{proof} messages to determine the safety of a leader's proposal. \Cref{alg:claim2/4} illustrates a function utilized in both \Cref{alg:leader finds v} and \Cref{alg:node determine value} describing the contents of \Cref{rule:claims_safe_leader} and \Cref{rule:claims_safe}. \Cref{rule:claims_safe_leader} and \Cref{rule:claims_safe} can be found in Appendix~\ref{sec:helper algo}.

\begin{algorithm} 
\caption{node\_claim\_safe($suggest/proof, v', val$)}
\begin{algorithmic}[1]\label{alg:claim2/4}
    \State \textbf{if} $suggest$ \textbf{then} $vote \gets suggest.vote\mathit{2}, prev\_vote \gets suggest.prev\_vote\mathit{2}$
    \State \textbf{if} $proof$ \textbf{then} $vote \gets proof.vote\mathit{1}, prev\_vote \gets proof.prev\_vote\mathit{1}$
    \If{$v' = 0$} 
        \State \Return $true$
    \ElsIf{$vote.view \ge v'$ \textbf{and} $vote.val = val$}
        \State \Return $true$
    \ElsIf{$prev\_vote.view \ge v'$}
        \State \Return $true$
    \Else 
        \State \Return $false$
    \EndIf
\end{algorithmic}
\end{algorithm}

The key point of \Cref{alg:leader finds v} lies in tracing back from view $v-1$ to find the view $v'$ as specified in \Cref{rule:picking-safe-proposal}, where $v$ denotes the current view. A notable optimization to reduce the algorithm's complexity lies in skipping a view if it lacks sufficient \texttt{vote-2} and \texttt{prev-vote-2} (see line~\ref{algo:2_19} in \Cref{alg:leader finds v}), as \Cref{rule:picking-safe-proposal} \Cref{case:case 3} is not satisfied. The overall computational complexity of \Cref{alg:leader finds v} is $O(v\times m\times n)$, where $m = O(n)$ is the number of possible values derived from \texttt{vote-3} and \texttt{prev-vote-2} messages, and $n$ is the total number of nodes.

The logic in lines~\ref{algo:3_11}-\ref{algo:3_19}  of \Cref{alg:node determine value} is similar to \Cref{alg:leader finds v}, but is simplified as the value is given. For verifying \Cref{rule:checking-safe-proposal} \Cref{case:two block sets}, the goal is to identify two blocking sets that claim two different values $\widetilde{val}$ and $\widetilde{val}'$ safe, at views $\tilde{v}$ and $\tilde{v}'$, satisfying $v'\le \tilde{v}<\tilde{v}'<v$. Although it involves identifying three specific views, $v'$, $\tilde{v}$ and $\tilde{v}'$, the process only backtracks from view $v-1$ once, which is crucial for complexity reduction. We first verify if the number of nodes claiming a value safe (using \Cref{alg:claim2/4}) meets the size of a blocking set. If so, we record the ($view$, $val$) pair, representing the possible blocking sets. We then verify every possible combination of these pairs to ensure that the two blocking sets are within the same quorum, thereby satisfying \Cref{rule:checking-safe-proposal} \Cref{{2bii}}. Upon determining $\tilde{v}$ and $\tilde{v}'$, it is critical to note that for identifying $v'$, verifying whether $v'=\tilde{v}$ meets \Cref{rule:checking-safe-proposal} \Cref{2bi} and \Cref{2bii} suffices. This is because if no quorum member sent any \texttt{vote-4} messages for a view strictly higher than $v'$ (\Cref{rule:checking-safe-proposal} \Cref{2bi}), this condition remains valid for smaller value of $v'$, thus also reducing the computational complexity. The algorithm maintains an overall complexity of $O(v\times m \times n)$, the same as \Cref{alg:leader finds v}.

\section{Security Analysis}
\label{sec:proof}
In this section, we prove that Basic TetraBFT solves the problem of consensus as defined in \Cref{def:concensus}. We first prove some useful lemmas on liveness and safety. The proof for liveness is the most interesting, and it follows the logic illustrated in Fig.~\ref{fig:lemma}. Our first objective is to prove the existence of a safe value, showing that a leader is capable of determining a value as safe upon receiving \texttt{suggest} messages from a quorum containing all well-behaved nodes. We then proceed to demonstrate that once a well-behaved node has received \texttt{proof} messages from each of the other well-behaved nodes, it determines the leader’s proposal safe. Finally, we show that when all well-behaved nodes determine the leader’s proposal safe, a decision is consequently made. 

The proof of safety is divided into two main components: within-view safety and cross-view safety. We first prove that well-behaved nodes cannot make conflicting decisions within the same view. Regarding cross-view safety, we show that after a value is decided by any well-behaved node, all well-behaved nodes will not send $\texttt{vote-1}$ messages for a different value in future views, thus preventing any possibility of conflicting decisions. 

Building on liveness and safety arguments, we finally present the main theorem. TetraBFT guarantees agreement, validity and termination in partial synchrony. 

\begin{figure}[t]
\centering
\includegraphics[height=2cm]{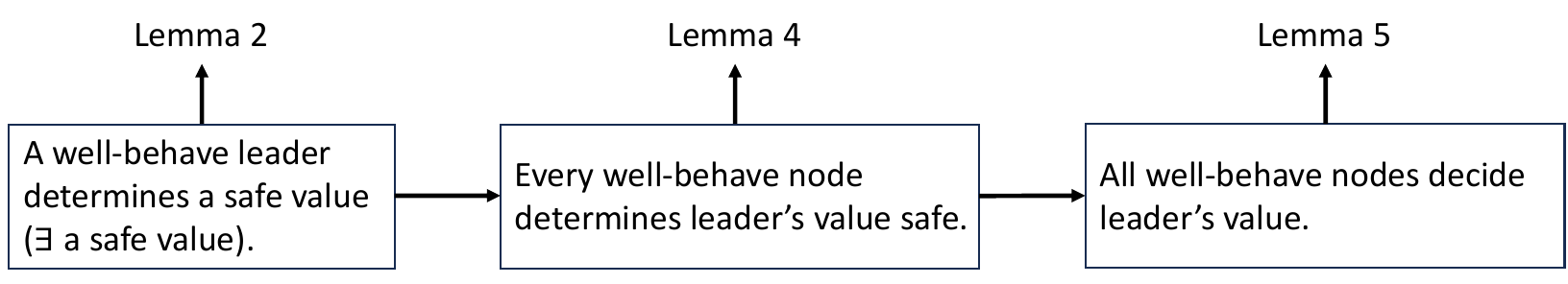}
\caption{Liveness lemmas logical framework.}
\label{fig:lemma}
\end{figure}

\subsection{Liveness Argument}\label{sec:Liveness Argument}

The argument concerning liveness is predicated on the assumption that view $v$ is led by a well-behaved leader and the network is synchronous (i.e.\ the view starts after GST). The first lemma establishes that when a well-behaved node sends a \texttt{vote-1} or a \texttt{vote-2} message for a value $val$, it means that the node claims $val$ is safe in all preceding views, and this assertion is maintained in all subsequent views. This lemma illustrates a property that the set of safe values claimed by a well-behaved node shrinks as views advance, yet the claim holds in later views. We prove the first statement of the lemma by considering two cases: whether the value of the highest \texttt{vote-1} message in later views matches $val$ or not. By mapping these two cases to \Cref{rule:claims_safe} \Cref{case:highest vote-1} and \Cref{rule:claims_safe} \Cref{case:second highest vote-1} respectively, the statement is proved. Similarly, by altering \texttt{vote-1} message to \texttt{vote-2} message and \texttt{proof} message to \texttt{suggest} message and using \Cref{rule:claims_safe_leader} \Cref{node claims safe in suggest message highest view} and \Cref{rule:claims_safe_leader} \Cref{node claims safe in suggest message second highest view}, the second statement is also validated.

\begin{lemma}\label{lemma: honest_vote_1_safe_in_proof}
  If a well-behaved node sends a \textup{\texttt{vote-1}} message for value $val$ in view $v$, it will claim $val$ is safe at any view $\le v$ in its \textup{\texttt{proof}} message in views greater than $v$. Similarly, if a well-behaved node sends a \textup{\texttt{vote-2}} message for value $val$ in view $v$, it will claim $val$ is safe at any view $\le v$ in its \textup{\texttt{suggest}} message in views greater than $v$.
\end{lemma}
\begin{proof}
  We prove the lemma for the case of the \texttt{proof} message and the same logic can be applied to the \texttt{suggest} message. Our aim is to demonstrate that if a well-behaved node, denoted as $i$, sends a \texttt{vote-1} message for value $val$ in view $v$, then in its \texttt{proof} message in views greater than $v$, it will claim $val$ is safe at any view $v'$ with $v'\le v$. Suppose $i$ sends a \texttt{vote-1} message for value $val$ in view $v$. In $i$'s \texttt{proof} message in a view greater than $v$, the highest \texttt{vote-1} message was sent in a view $v''$ with $v'' \geq v \geq v'$ for some value $val'$. If $val' = val$, then the highest \texttt{vote-1} message was sent at a view $v''\ge v'$ and for value $val$, satisfying \Cref{rule:claims_safe} \Cref{case:highest vote-1}. If $val' \neq val$, then the second highest view for which $i$ sent a \textup{\texttt{vote-1}} message (for a different value from the highest \texttt{vote-1}) was at a view $\geq v'$, satisfying \Cref{rule:claims_safe} \Cref{case:second highest vote-1}. Combining these two cases, we have that node $i$ claims in its \texttt{proof} message in views greater than $v$ that $val$ is safe at view $v'$. Similarly, by altering \texttt{vote-1} message to \texttt{vote-2} message and \texttt{proof} message to \texttt{suggest} message, using \Cref{rule:claims_safe_leader} \Cref{node claims safe in suggest message highest view} and \Cref{rule:claims_safe_leader} \Cref{node claims safe in suggest message second highest view}, the same line of reasoning substantiates the second statement. Thus, the lemma is proven.
\end{proof}

The following lemma proves that a leader is able to determine a safe value once it receives \texttt{suggest} messages from a quorum containing all well-behaved nodes, while \Cref{alg:leader finds v} gives an efficient solution to find the safe value. We validate this lemma by analyzing \Cref{rule:picking-safe-proposal}. The case corresponding to \Cref{rule:picking-safe-proposal} \Cref{case:no_vote_leader} is straightforward. Regarding \Cref{rule:picking-safe-proposal} \Cref{case:highest_vote_leader}, we consider the last view $v'$ where some node in the quorum sent \texttt{vote-3} messages, as specified in \Cref{rule:picking-safe-proposal} \Cref{case:case 1}. Building on this, the fulfillment of \Cref{rule:picking-safe-proposal} \Cref{case:case 2} can be easily inferred by the quorum-intersection property (i.e., the intersection of two quorums must contain at least one well-behaved node), and \Cref{rule:picking-safe-proposal} \Cref{case:case 3} is logically supported by \Cref{lemma: honest_vote_1_safe_in_proof}.
\begin{lemma}
  Upon receiving \textup{\texttt{suggest}} messages from a quorum containing all well-behaved nodes, a leader determines some value $val$ that is safe.
  \label{claim:safety existence}
\end{lemma}

\begin{proof}
  Suppose in a view $v\neq 0$, a well-behaved leader received \texttt{suggest} messages from a quorum $q$ containing all well-behaved nodes. If no member of $q$ sent any \texttt{vote-3} messages before view $v$, then any value proposed by the leader is safe according to \Cref{rule:picking-safe-proposal} \Cref{case:no_vote_leader}. Now consider there is a view $v'<v$ which is the last view in which some nodes in $q$ sent \texttt{vote-3} messages. Suppose a member of quorum $q$, node $i$, has sent a \texttt{vote-3} message for value $val$ in view $v'$. In this case, node $i$ must have received \texttt{vote-2} messages for value $val$ in view $v'$ from a quorum of nodes, within which a blocking set $b$ of nodes are well-behaved. According to \Cref{lemma: honest_vote_1_safe_in_proof}, we have that all nodes in the blocking set $b$ claim in their \texttt{suggest} messages that $val$ is safe at $v'$. In addition, since no member of $b$ sent \texttt{vote-2} message for any value $val' \neq val$ in view $v'$, then there was no quorum of nodes that sent \texttt{vote-2} message for $val'$ in $v'$. Thus, no member in $q$ sent a \texttt{vote-3} message for any value other than $val$ in $v'$. Then, by \Cref{rule:picking-safe-proposal} \Cref{case:highest_vote_leader}, the claim is proved. 
\end{proof}

The subsequent lemma outlines the implications arising when a well-behaved node claims a value safe in its \texttt{suggest} message. Specifically, it shows how \Cref{rule:checking-safe-proposal} derives from each item in \Cref{rule:claims_safe_leader}. The case for \Cref{rule:claims_safe_leader} \Cref{view 0} is straightforward. By applying \Cref{lemma: honest_vote_1_safe_in_proof}, we show that \Cref{rule:checking-safe-proposal} \Cref{2bii} and \Cref{rule:checking-safe-proposal} \Cref{2biii} directly relate to \Cref{rule:claims_safe_leader} \Cref{node claims safe in suggest message highest view} and \Cref{rule:claims_safe_leader} \Cref{node claims safe in suggest message second highest view}, respectively. 
\begin{lemma}
  \label{claim:blocking_claims_safe}
  If a well-behaved node claims that $val$ is safe at $v'$ in its \textup{\texttt{suggest}} message in view $v > v'$,  
  \begin{enumerate}
      \item then there is a blocking set $b$ composed entirely of well-behaved nodes such that, in any \textup{\texttt{proof}} message in views greater than or equal to $v$, every member of $b$ claims that $val$ is safe at $v'$,
      \item or there are two blocking sets $\tilde{b}$ and $\tilde{b}'$ both composed entirely of well-behaved nodes such that, in any \textup{\texttt{proof}} message in views greater than or equal to $v$, every member of $\tilde{b}$ claims that $\widetilde{val}$ is safe at $\tilde{v}$ where $v'\le \tilde{v} < v$, and every member of $\tilde{b}'$ claims that $\widetilde{val}'\neq \widetilde{val}$ is safe at $\tilde{v}'$ where $\tilde{v}<\tilde{v}'<v$.
  \end{enumerate} 
\end{lemma}

\begin{proof}
  Suppose a well-behaved node $i$ claims value $val$ is safe at view $v'$ in its \texttt{suggest} message in view $v$. By \Cref{rule:claims_safe_leader}, there are three cases and we will discuss each one separately. 
  
  First (\Cref{rule:claims_safe_leader} \Cref{view 0}), if $v'=0$, then \Cref{rule:claims_safe} \Cref{case:r'=0} will trivially be true for every well-behaved node. 
  
  Second (\Cref{rule:claims_safe_leader} \Cref{node claims safe in suggest message highest view}), if node $i$'s highest \texttt{vote-2} message was sent at view $v''\ge v'$ for value $val$, then it has received a quorum of \texttt{vote-1} messages for value $val$ in view $v''$. Since $3f<n$, it can be inferred that there is a well-behaved blocking set $b$ that has sent \texttt{vote-1} messages in view $v''$ for value $val$.  According to \Cref{lemma: honest_vote_1_safe_in_proof}, it follows that every node in the blocking set $b$ claims in its \texttt{proof} message in views greater than or equal to $v$ that $val$ is safe at view $v'$, which proves the claim by satisfying the first consequent.
  
  Third (\Cref{rule:claims_safe_leader} \Cref{node claims safe in suggest message second highest view}), node $i$'s second highest \texttt{vote-2} message was sent at view $\tilde{v}$ with $v'\le \tilde{v} < v$ for value $\widetilde{val}$, and its highest \texttt{vote-2} message was sent at a view denoted as $\tilde{v}'$, with $\tilde{v}<\tilde{v}'<v$ and for value $\widetilde{val}'\neq \widetilde{val}$. From $i$'s second highest \texttt{vote-2} message, it can be deduced that in view $\tilde{v}$ a quorum of nodes has sent \texttt{vote-1} messages for value $\widetilde{val}$, where a blocking set $\tilde{b}$ of nodes are well-behaved. According to \Cref{lemma: honest_vote_1_safe_in_proof}, it follows that in any \textup{\texttt{proof}} message in views greater than or equal to $v$, every member of $\tilde{b}$ claims that $\widetilde{val}$ is safe at $\tilde{v}$. Similarly, from $i$'s highest \texttt{vote-2} message, in view $\tilde{v}'$ a quorum of nodes sent \texttt{vote-1} messages for value $\widetilde{val}'$, where a blocking set $\tilde{b}'$ of nodes are well-behaved. As per \Cref{lemma: honest_vote_1_safe_in_proof}, it holds that in any \textup{\texttt{proof}} message in views greater than or equal to $v$, every member of $\tilde{b}'$ claims that $\widetilde{val}'$ is safe at $\tilde{v}'$. Thus, the second consequent of \Cref{claim:blocking_claims_safe} is satisfied and we can conclude the proof.
\end{proof}

The following lemma shows that a well-behaved node will determine a well-behaved leader's proposal as safe. To prove this lemma, we demonstrate alignment between the cases outlined in \Cref{rule:picking-safe-proposal} and those in \Cref{rule:checking-safe-proposal}. The cases corresponding to $v=0$ and \Cref{rule:picking-safe-proposal} \Cref{case:no_vote_leader} are straightforward. As for the case corresponding to \Cref{rule:picking-safe-proposal} \Cref{case:highest_vote_leader}, we establish a one-to-one mapping from \Cref{rule:picking-safe-proposal} \Cref{case:case 1}, \Cref{case:case 2}, \Cref{case:case 3} to \Cref{rule:checking-safe-proposal} \Cref{2bi}, \Cref{2bii}, \Cref{2biii} by quorum-intersection and \Cref{{claim:blocking_claims_safe}}. 

\begin{lemma}
  \label{claim:proposal_determined_safe}
  A well-behaved node eventually determines that the leader's value $val$ is safe.
\end{lemma}
\begin{proof}
  According to the rule that a well-behaved leader uses to propose a safe value (\Cref{rule:picking-safe-proposal}), there are three cases. 
  
  First, if $v = 0$, then all well-behaved nodes trivially determine that the leader's value is safe. 

  Second (\Cref{case:no_vote_leader}), suppose that $v \neq 0$ and that the leader proposes $val$ because a quorum $q$ reports not sending any \texttt{vote-3} messages. Then, there is an entirely well-behaved blocking set $b$ that never sent any \texttt{vote-3} messages. Since \texttt{vote-4} messages are sent in response to a quorum of \texttt{vote-3} messages, and since a quorum and a blocking set must have a well-behaved node in common, we conclude that no well-behaved node ever sent a \texttt{vote-4} message. Thus, once a well-behaved node $i$ receives \texttt{proof} messages from all other well-behaved nodes, $i$ concludes that the proposal is safe according to \Cref{rule:checking-safe-proposal} \Cref{case:no_votes}.

  Third (\Cref{case:highest_vote_leader}), suppose that the view is not 0 and we have a quorum $q$ and a view $v'< v$ such that:
  \begin{enumerate}
    \item no member of $q$ sent any \texttt{vote-3} messages for a
      view strictly higher than $v'$, and
      \label{item:one}
    \item any member of $q$ that sent a \texttt{vote-3} message in
      view $v'$ did so with value $val$, and
      \label{item:two}
    \item there is a blocking set $b$ (e.g. $f+1$
      nodes) that all claim in their \texttt{suggest} messages that
      $val$ is safe at $v'$ (see \Cref{rule:claims_safe_leader}).
      \label{item:three}
  \end{enumerate}

  We make the following observations:
  \begin{enumerate}
    \item[a)] By \Cref{item:one} above, no well-behaved node sent any \texttt{vote-4} message in any view higher than $v'$; otherwise, a quorum would have sent the corresponding \texttt{vote-3} messages and, by the quorum-intersection property, this contradicts \Cref{item:one}.
    \item[b)] By \Cref{item:two} above, we can deduce any well-behaved node that sent a \texttt{vote-4} message in view $v'$ did so for value $val$. This can be proved by contradiction. Suppose there is a well-behaved node sent a \texttt{vote-4} message for value $val'\neq val$ in view $v'$. In this case, it implies that a well-behaved blocking set $b$ sent \texttt{vote-3} message for value $val'$ in view $v'$. According to \Cref{item:two} above, and using the quorum-intersection property, the quorum $q$ with $val$ intersect with blocking set $b$ with $val'$, and this intersection includes at least one well-behaved node. This situation leads to a contradiction.  
    \item[c)] By \Cref{item:three}, there is a well-behaved node that claims that $val$ is safe in $v'$ in its \texttt{suggest} message. By \Cref{claim:blocking_claims_safe}, we conclude that there is a blocking set $b$ composed entirely of well-behaved nodes that claim in their \texttt{proof} messages that $val$ is safe at $v'$, or there are two blocking sets $\tilde{b}$ and $\tilde{b}'$ both composed entirely of well-behaved nodes such that, every member of $\tilde{b}$ claims in its \texttt{proof} message that $\tilde{val}$ is safe at $\tilde{v}$ where $v'\le \tilde{v} < v$, and every member of $\tilde{b}'$ claims in its \texttt{proof} message that $\tilde{val}'\neq \tilde{val}$ is safe at $\tilde{v}'$ where $\tilde{v}<\tilde{v}'<v$.
  \end{enumerate}
  By Items a), b), and c) and \Cref{rule:checking-safe-proposal}, we conclude that, once every well-behaved node has received a \texttt{proof} message from every other well-behaved node, every well-behaved node determines that the leader's proposal is safe.
\end{proof}

The final lemma on liveness shows if the leader's proposal is determined safe by all the well-behaved nodes, then the protocol will proceeds smoothly and a decision will be reached by all the well-behaved nodes.

\begin{lemma}
  \label{claim:if_safe_then_termination}
  If all well-behaved nodes determine that the leader's value is safe, then a
  decision is made by all well-behaved nodes.
\end{lemma}
\begin{proof}
  If all well-behaved nodes determine that the leader's value $val$ is safe, they will broadcast \texttt{vote-1} messages for $val$. Consequently, all the well-behaved nodes receive $n-f$ \texttt{vote-1} messages for $val$  and broadcast \texttt{vote-2} messages for $val$. Continuing in the same logic, all well-behaved nodes broadcast \texttt{vote-3} and \texttt{vote-4} messages for $val$. Upon receiving $n-f$ \texttt{vote-4} messages for $val$, a decision is made by all well-behaved nodes.
\end{proof}

\subsection{Safety Argument}
\label{sec:safety}

The first two lemmas show that within a view, well-behaved nodes send \textup{\texttt{vote}} messages for the same value and decide the same value. These lemmas follow directly from the evolution of a view as discussed in Section~\ref{sec:evolution-view} and can primarily be proven through the quorum-intersection property.
\begin{lemma}
\label{safety corollary}
  If two well-behaved nodes send a \textup{\texttt{vote}} message in view $v$, where \textup{\texttt{vote}} can be one of \textup{\{\texttt{vote-2}},\textup{\texttt{vote-3},\texttt{vote-4}\}}, for values $val$ and $val'$ separately, then $val=val'$.
\end{lemma}
\begin{proof}
  Observe two well-behaved nodes $i$ and $j$ send a \texttt{vote-2} message in view $v$ for $val$ and $val'$ respectively. Then, node $i$ received \texttt{vote-1} messages in view $v$ for $val$ from a quorum of nodes $q_1$, and node $j$ received \texttt{vote-1} messages in view $v$ for $val'$ from a quorum of nodes $q_2$. Since there are only $n$ nodes, $q_1 \cap q_2$ must contain at least one well-behaved node. Therefore, $val=val'$ as well-behaved nodes only send one \texttt{vote-1} message in a view. By similar logic, the same result can be concluded for \texttt{vote-3} and \texttt{vote-4}.
\end{proof}

\begin{lemma} \label{lemma: safety inner view}
  If a well-behaved node decides a value $val$ in view $v$, then no well-behaved node decides other values in view $v$.
\end{lemma}
\begin{proof}
  In view $v$, if a well-behaved node decides value $val$, it must have first received \texttt{vote-4} messages for value $val$ from a quorum of nodes. By \Cref{safety corollary}, all well-behaved nodes that send 
  \texttt{vote-4} messages in view $v$ do so with the same value $val$. Well-behaved nodes only send \texttt{vote-4} messages for a value $val$ after receiving at least $n-f$ \texttt{vote-3} messages for the value $val'$. Since $n-f>f$, at least one of those messages must have been received from a well-behaved node, and thus $val'=val$. By applying similar reasoning, every well-behaved node that sends a \texttt{vote-1} message does so with value $val''=val$. If two well-behaved nodes separately decide different values, it implies that two quorums of nodes must have sent \texttt{vote-1} messages for two different values, leading to a contradiction.
\end{proof}

In the following lemma, we demonstrate that if a value is decided in view $v$, then in  later views, no well-behaved node determines a different value as safe. We approach the proof by contradiction, starting with the assumption that there exists a scenario where a different value is determined safe (in accordance with \Cref{rule:checking-safe-proposal}), focusing on the first view where this occurs. Then we compare view $v$ with the view appeared in the \Cref{rule:checking-safe-proposal} \Cref{2bi} (denoted as $v''$), leading to three cases. It becomes straightforward to check that cases where $v''<v$ and $v''=v$ result in contradictions due to the quorum-intersection property. Then we discuss the case where $v''>v$, examining \Cref{case:one block set} and \Cref{case:two block sets}. For \Cref{case:one block set}, we verify the cases outlined in \Cref{rule:claims_safe} and find that, under the assumption that this is the first view a different value is determined safe, they either do not exist or lead to contradictions. In the case of \Cref{case:two block sets}, we establish that it is impossible for the two blocking sets to meet any of the four possible combinations of \Cref{rule:claims_safe} \Cref{case:second highest vote-1} and \Cref{rule:claims_safe} \Cref{case:highest vote-1} by the logic of views and values. Additionally, it is straightforward to verify that \Cref{rule:claims_safe} \Cref{case:r'=0} cannot hold.
\begin{lemma} \label{lemma: safety across view}
  If a well-behaved node decides a value $val$ in view $v$, then no well-behaved node sends a \textup{\texttt{vote-1}} message for a different value $val'$ of any later view $v'>v$.
\end{lemma}
\begin{proof}
  Assume by contradiction that a well-behaved node $i$ decides value $val$ in view $v$, and let $v' > v$ be the first view in which some well-behaved node (denoted as $j$) sends \texttt{vote-1} message for a value other than $val$ (denoted as $val'$). Before sending \texttt{vote-1}, node $j$ receives a quorum $q'$ of \texttt{proof} messages in view $v'$, and determines the proposal from the leader for value $val'$ is safe. As node $i$ decides value $val$ in view $v$, there must be a quorum $q$ that has sent \texttt{vote-4} messages in view $v$. Then we discuss cases of \Cref{rule:checking-safe-proposal} \Cref{case:no_votes} and \Cref{rule:checking-safe-proposal} \Cref{case:highest_vote} individually. 
  \begin{enumerate}
  \item[1.]
  Since $q \cap q'$ must contain one well-behaved node who has sent \texttt{vote-4} in view $v$, the case of \Cref{rule:checking-safe-proposal} \Cref{case:no_votes} is impossible. 
  \item[2.]
  Now consider \Cref{rule:checking-safe-proposal} \Cref{case:highest_vote}, and there is a view $v''<v'$ satisfies conditions in \Cref{rule:checking-safe-proposal} \Cref{case:highest_vote}. In this case, we discuss three cases separately: $v''<v$, $v''=v$, and $v''>v$. 
   \begin{enumerate}
     \item[a)]
     When $v'' < v$, a quorum $q$ sends \texttt{vote-4} messages in view $v$. Two quorums $q$ and $q'$ intersect at least one well-behaved node, who has sent a \texttt{vote-4} message in view $v>v''$,  contradicting \Cref{rule:checking-safe-proposal} \Cref{2bi}. 
     \item[b)]
     When $v'' = v$, again by quorum intersection, there must exists at least one well-behaved node in $q'$ that has sent a \texttt{vote-4} message for value $val$ in view $v''=v$; however, any member of $q'$ should not have sent \texttt{vote-4} message for any value other than $val'$ according to \Cref{rule:checking-safe-proposal} \Cref{2bii}. This is a contradiction.
     \item[c)]
     When $v'' > v$, if \Cref{rule:checking-safe-proposal} \Cref{case:one block set} is satisfied, there is a blocking set $b$ where all members claim in their \texttt{proof} messages that $val'$ is safe at $v''$. Notably, $b$ must contain at least one well-behaved node, which we denote as $k$. Subsequently, we examine each condition in \Cref{rule:claims_safe} for $k$.
     \begin{enumerate}
      \item[i.]
      Since $v''>v$, then $v''\neq 0$, making it impossible for \Cref{rule:claims_safe} \Cref{case:r'=0} to hold. 
      \item[ii.]
      As for the circumstance of \Cref{rule:claims_safe} \Cref{case:highest vote-1}, if $val'$ is claimed safe at $v''$ in $k$'s \texttt{proof} message, there is a view $v'''$ with $v''\le v''' < v'$ where $k$ sends a \texttt{vote-1} message for value $val'$. However, $v'$ is the first view where some node sends \texttt{vote-1} for $val'$. Thus, this situation is nonexistent.
      \item[iii.]
      According to \Cref{rule:claims_safe} \Cref{case:second highest vote-1}, $k$ sends a \texttt{vote-1} message for value $\widebar{val}$ in view $\bar{v}$ and a \texttt{vote-1} message for $\widebar{val}'$ in view $\bar{v}'$, where $v''\le \bar{v}<\bar{v}'<v'$ and $\widebar{val}\neq \widebar{val}'$. Since $v'$ is the first view where a \texttt{vote-1} message for a value other than $val$ is sent, we have $\widebar{val}=\widebar{val}'=val$. This is a contradiction.
     \end{enumerate}
     \item[d)]
      When $v'' > v$, if \Cref{rule:checking-safe-proposal} \Cref{case:two block sets} is satisfied, there is a blocking set $\tilde{b}$ where all members claim in their \texttt{proof} messages that $\widetilde{val}$ is safe at $\tilde{v}$ where $v'' \le \tilde{v}<v'$, and there is a blocking set $\tilde{b}'$ where all members claim in their \texttt{proof} messages that $\widetilde{val}'\neq \widetilde{val}$ is safe at $\tilde{v}'$ where $\tilde{v}< \tilde{v}'<v'$. Notably, $\tilde{b}$ must contain at least one well-behaved node, denoted as $m$, and similarly, $\tilde{b}'$ must contain at least one well-behaved node, denoted as $m'$. We discuss the scenarios involving combinations of \Cref{rule:claims_safe} for $m$ and $m'$.
        \begin{enumerate}
          \item[i.]
            Since $v''>v$, then $v''\neq 0$, making it impossible for \Cref{rule:claims_safe} \Cref{case:r'=0} to hold for either $m$ and $m'$. 
          \item[ii.]
            Next, we examine the scenario in which either of $m$ or $m'$ satisfies \Cref{rule:claims_safe} \Cref{case:second highest vote-1}. Without loss of generality, we assume $m$ satisfies \Cref{rule:claims_safe} \Cref{case:second highest vote-1}. Accordingly, $m$ sends a \texttt{vote-1} message for value $\widebar{val}$ in view $\bar{v}$ and a \texttt{vote-1} message for $\widebar{val}'\neq \widebar{val}$ in view $\bar{v}'$, where $\tilde{v}\le \bar{v}<\bar{v}'<v'$. Since $v'$ is the first view wherein a \texttt{vote-1} message for a value different from $val$ is sent, we deduce that $\widebar{val}=\widebar{val}'=val$. Thus, this situation does not exist. Likewise, the same holds true for node $m'$. Therefore, either $m$ or $m'$ satisfies \Cref{rule:claims_safe} \Cref{case:second highest vote-1}.
          \item[iii.]
            Then, only one case remains, that is, both $m$ and $m'$ follows \Cref{rule:claims_safe} \Cref{case:highest vote-1}. Accordingly, since $\widetilde{val}$ is claimed safe at $\tilde{v}$ in $m$'s \texttt{proof} message, there is a view $\bar{v}$ with $\tilde{v}\le \hat{v} < v'$ where $m$ sends a \texttt{vote-1} message for value $\widetilde{val}$. Since $v'$ is the first view where some node sends \texttt{vote-1} for a value different from $val$. Then, we can deduce that $\widetilde{val}=val$. Similarly, we have $\widetilde{val}'=val$. However, this contradicts the premise of $\widetilde{val}\neq \widetilde{val}'$. Consequently, this situation is nonexistent.
      \end{enumerate}
    \end{enumerate}
  \end{enumerate}
     Thus, in no circumstances does node $j$ decide that leader's proposal for value $val'$ safe. Therefore, it never sends a \texttt{vote-1} message for a value $val'$ in view $v'>v$, which concludes the claim.
\end{proof}

\subsection{Main Theorem}
Building on these arguments, we now present the main theorem. 

\begin{theorem}
\label{thm:main}
TetraBFT guarantees agreement, validity and termination in partial synchrony.  Specifically, it guarantees agreement and validity even in periods of asynchrony, while termination is assured after GST.
\end{theorem}

\begin{proof}
The proof for each property is provided separately.

\textbf{Termination.} By \Cref{claim:safety existence}, \Cref{claim:proposal_determined_safe} and \Cref{claim:if_safe_then_termination}, we conclude that a decision is reached by all well-behaved nodes. 

\textbf{Agreement.} We consider two scenarios involving well-behaved nodes $i$ and $j$ deciding value $val$ and $val'\neq val$ respectively. If $i$ and $j$ decide value in the same view, \Cref{lemma: safety inner view} implies that $val = val'$, as a contradiction would arise otherwise. Now, suppose $i$ decides $val$ $j$ in view $v$ and $j$ decides value $val'\neq val$ in a later view $v'>v$. According to \Cref{lemma: safety across view}, no well-behaved node would send a \texttt{vote-1} message with value different from $val$ in view $v'$ or subsequent views. Similarly, no well-behaved node would send a \texttt{vote-2} message for a different value $val'$ in any $v'>v$. That is because, to do so, it must have first received $n-f$ \texttt{vote-1} messages for value $val'$. For the same reasons, no well-behaved node sends \texttt{vote-3} or \texttt{vote-4} messages for a value $val'\neq val$ and ultimately decides it in view $v'$. Based on these analyses, it's evident that once a value $val$ is decided by a well-behaved node, no other well-behaved node can decide a different value $val'\neq val$ at any stage of the protocol. 

\textbf{Validity.} Assume that all nodes are well-behaved and that they have the same input value $val$. Consider a well-behaved leader proposes the initial value $val$, and each well-behaved node sends vote messages do so for value $val$ in the previous views. Consequently, by assumption, the initial value $val$ is determined safe. Then \texttt{vote-1} message sent by well-behaved nodes will be for $val$. If a node sends a \texttt{vote-2} messages for a different value $val'$, it must have first received $n-f$ \texttt{vote-1} messages for $val'$. However, since all \texttt{vote-1} messages from well-behaved nodes are for $val$, it logically follows that $val'=val$. Subsequently, if a node sends \texttt{vote-3} messages for a value $val'$, it must have received $n-f$ \texttt{vote-2} messages for value $val'$. However, since two quorums intersect at least one behave node, it is impossible that a quorum of nodes send \texttt{vote-2} for $val$ and a quorum of nodes send \texttt{vote-2} for a value $val'\neq val$, given that some well-behaved node has sent a \texttt{vote-2} message for $val$. Similarly, if a node sends \texttt{vote-4} and decides a value, that value cannot be any value other than $val$.
\end{proof}

\section{formal verification}

To gain one more degree of confidence in TetraBFT, we formalize the single-shot protocol of~\Cref{sec:single-shot} in TLA+~\cite{lamport_specifying_2002} and we use the Apalache model-checker~\cite{konnov_tla_2019} to formally verify that the agreement property holds in all possible executions of a system with 4 nodes, including one Byzantine node, that try to agree on a value among 3 different values and execute 5 views.

Since every view consists of 4 voting phases, including Byzantine behavior, the size of the state-space to explore, even within those seemingly small bounds, is staggeringly large.
In our experiments, explicitly exploring even a small fraction of this state-space was out of reach with explicit state exploration using the TLC model-checker~\cite{lamport_specifying_2002}.
Checking the agreement property by unrolling the transition relation, which is the default mode with the Apalache model-checker, is also unrealistic because, given our model of the protocol, Apalache would need to unroll a complex transition relation at least $5*4*4=80$ times (5 views times 4 voting phases times 4 nodes) to reach the end of view 5.
In our experiments, verifying the agreement property when the transition relation is unfolded only 10 times already takes too long.

Instead, to perform the verification, we provide the Apalache model-checker with a candidate inductive invariant $Inv$ implying the agreement property, and we ask Apalache to verify that the invariant is inductive.
This involves checking that $Inv$ holds in the initial state and checking that if $Inv$ holds in an arbitrary state and the system takes one single step, then $Inv$ holds again in the new state.
For a symbolic model-checker like Apalache, this is much easier than unrolling the transition relation, and Apalache successfully verifies that the invariant is inductive in about three hours on a consumer desktop machine.
The formal specification appears in full in~\Cref{sec:formal-spec} and are also available on Zenodo at \url{https://doi.org/10.5281/zenodo.19321940}.

\section{Multi-Shot TetraBFT}
\label{sec:multi}
In this section, we extend Basic TetraBFT, which enables nodes to reach consensus on a single value, to Multi-shot TetraBFT. This extension allows nodes to achieve consensus on a sequence of blocks, thereby forming a blockchain. We first demonstrate the protocol in the good case. Subsequently, building on the protocol for the ideal case, we explore scenarios involving timeouts, thereby necessitating the incorporation of a view-change mechanism.

\subsection{Multi-Shot TetraBFT in the Good Case}
In the Multi-shot TetraBFT, blocks are indexed by slot numbers. A pre-determined leader in each slot appends its block to that of the previous slot. Upon receiving the block $b_i$ for slot $i$, each node ensures that 1) a block $b_{i-1}$ for slot $i-1$ has received a quorum of votes; 2) $b_i$ extends $b_{i-1}$. Once both conditions are satisfied, a well-behaved node then broadcasts a \texttt{vote} message for $b_i$. Simultaneously, upon receiving block $b_i$ and confirming both conditions, a well-behaved leader for slot $i+1$ proposes a new block extending $b_i$. This proposal can effectively serve as an implicit vote to save one message per slot. 
A block is \textit{notarized} on receiving votes from a quorum of nodes. The first block in a chain of four notarized blocks with consecutive slot numbers is finalized, as well as its entire prefix in the chain.

In Fig.~\ref{fig:view0}, we give an example to illustrate the protocol under the good case (with well-behaved leaders and synchrony). Starting with view 0 of slot $s$, a well-behaved leader proposes a block with value $val_1$. On receiving $val_1$, a well-behaved leader proposes a block with value $val_2$. At the same time, each well-behaved node determines $val_1$ safe (as the view number is 0), and broadcasts a \texttt{vote} message $\langle \texttt{vote}, slot-s, view-0, value-val_1\rangle$, which represents \texttt{vote-1} for $val_1$.
In slot $s+1$, upon determining $val_2$ safe and notarizing the block for slot $s$ (i.e., receiving a quorum of \texttt{vote} messages for it), each well-behaved node broadcasts a \texttt{vote} message $\langle \texttt{vote}, slot-s+1, view-0, value-val_2\rangle$. This vote represents \texttt{vote-1} for $val_2$ and \texttt{vote-2} for $val_1$. 
Simultaneously, a well-behaved leader proposes a block with value $val_3$.
In slot $s+2$, after determining $val_3$ safe and receiving a quorum of \texttt{vote} messages for $val_2$ of slot $s+1$, a well-behaved node broadcasts a \texttt{vote} message $\langle \texttt{vote}, slot-s+2, view-0, value-val_3\rangle$. Similarly, this vote represents \texttt{vote-1} for $val_3$, \texttt{vote-2} for $val_2$ and \texttt{vote-3} for $val_1$. The process in slot $s+3$ mirrors that of slot $s+2$, hence its explanation is omitted here. After nodes receiving a quorum of \texttt{vote} messages $\langle \texttt{vote}, slot-s+3, view-0, value-val_4\rangle$, the block for slot $s+3$ is notarized and the block for slot $s$ is finalized.
\begin{figure}[hbt!]
\centering
\includegraphics[height=4.7cm]{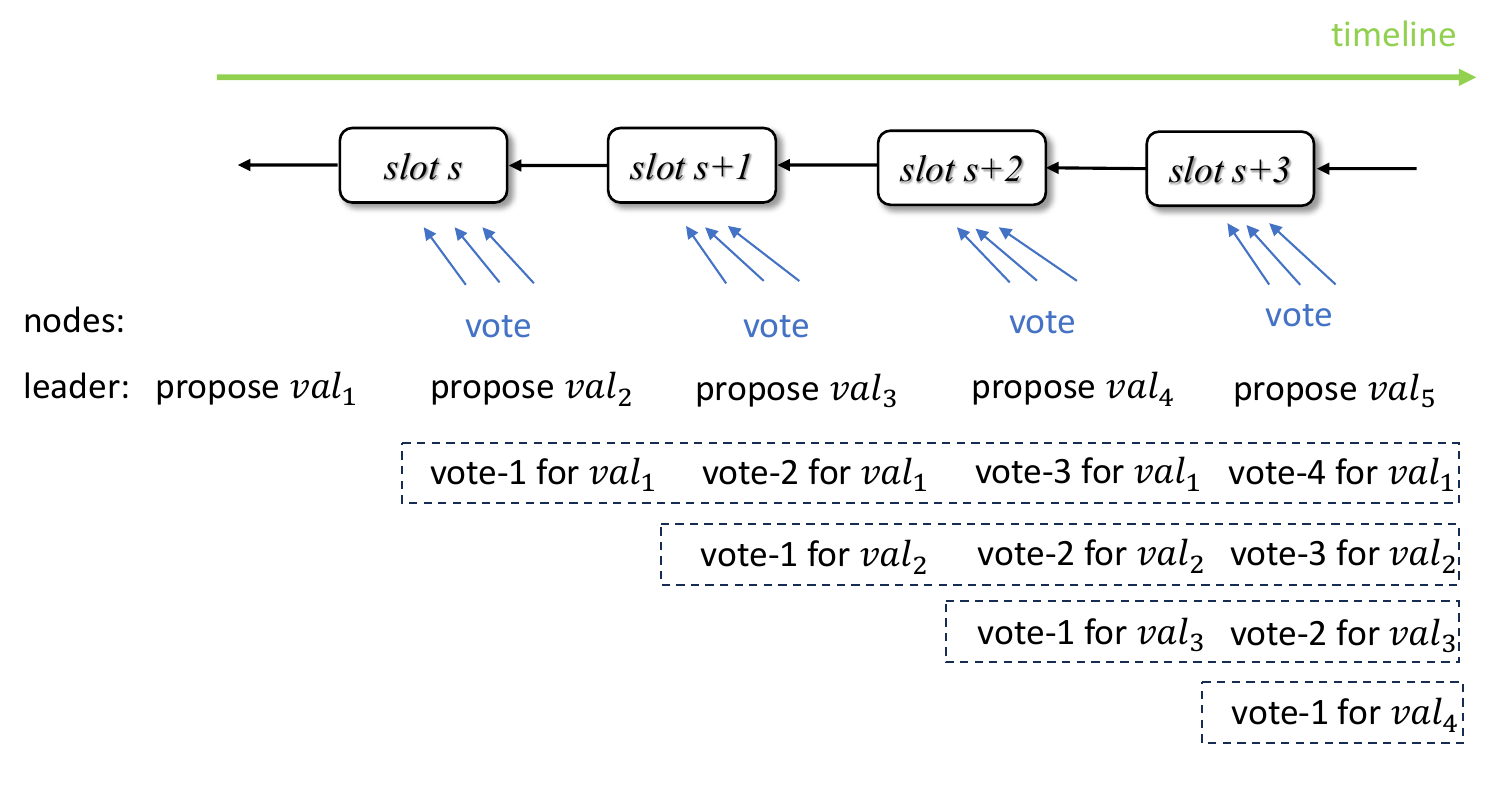}
\caption{Example of Multi-shot TetraBFT in the good case.}
\label{fig:view0}
\end{figure}

\subsection{View Change in Multi-Shot TetraBFT}

The preceding section presents an ideal scenario where a block is successfully finalized in view 0. In this section, we illustrate the  protocol's operation in instances requiring a view change due to aborted blocks (refer to \Cref{alg:view change} and \ref{alg:chained protocol}). Upon receiving a \texttt{proposal} message for the slot $s-1$, a well-behaved node starts the slot $s$ and sets the view timer. Should the $9\Delta$ timeout elapse without block finalization, the node broadcasts a \texttt{view-change} message for the next view, indicating the lowest slot number of the aborted blocks. Note that if a block is aborted, all subsequent blocks are aborted as well; however, the number of aborted blocks is limited by the protocol's finality latency, specifically to 5.
On receiving $f+1$ \texttt{view-change} messages for view $v$ and slot $s$, a node broadcasts a \texttt{view-change} message for view $v$ and slot $s$ if it has not done so for view $v$ or any higher view in slot $s$.  On receiving $n-f$ \texttt{view-change} messages for view $v$ and slot $s$, a node changes to view $v$ for all slots numbered no less than $s$ and resets the corresponding timers. Nodes keep monitoring \texttt{view-change} messages regardless of the current view.
Upon view change, a node sends \texttt{suggest}/\texttt{proof} messages for all aborted slots.
On receiving the \texttt{suggest} messages, a well-behaved leader of view $v$ and slot $s$ proposes a new block with a safe value. Subsequently, the protocol mirrors the good case, with the distinction that proposals and votes for safe values are informed by \texttt{suggest}/\texttt{proof} messages and guided by \Cref{rule:picking-safe-proposal} and \Cref{rule:checking-safe-proposal} in views other than 0.
\begin{figure}[h]
\centering
\includegraphics[height=5.0cm]{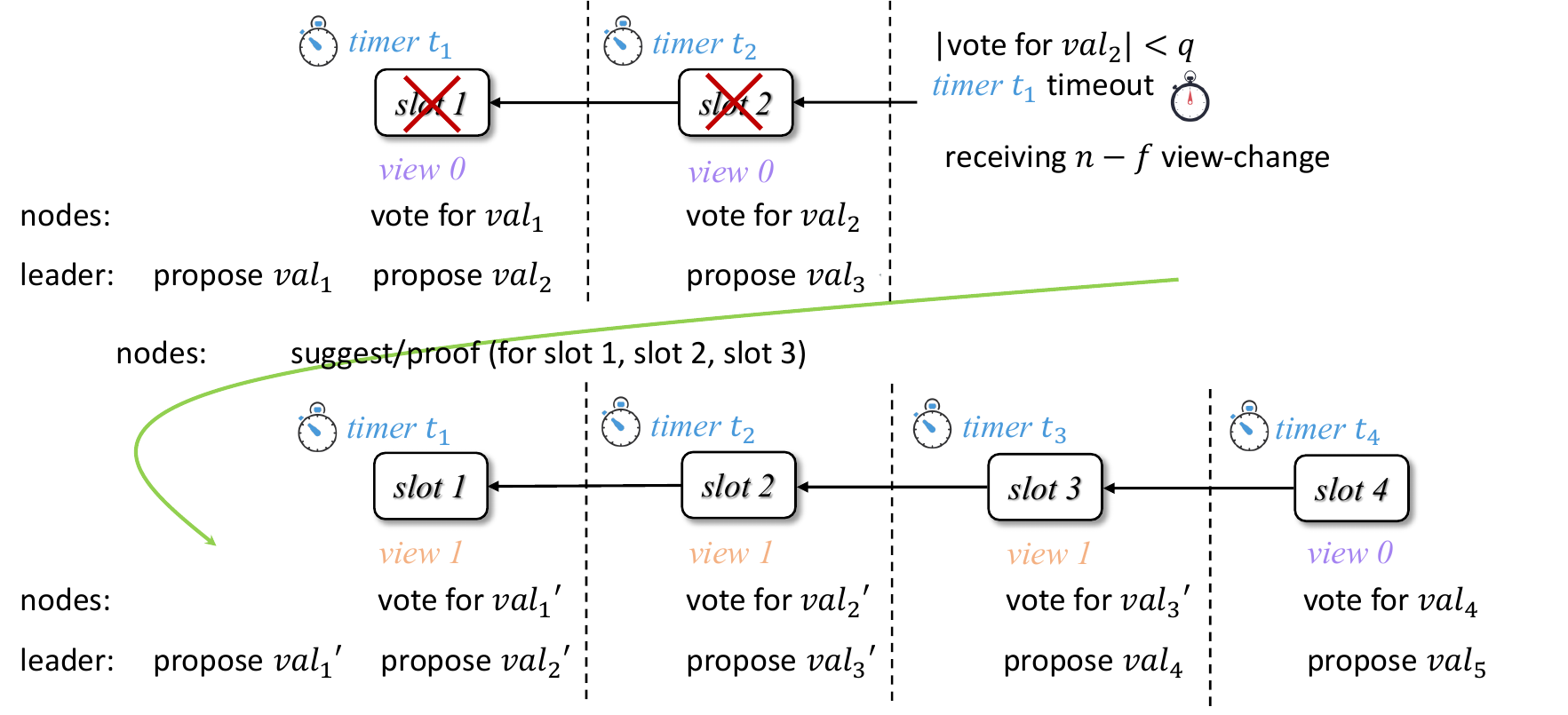}
\caption{Example of Multi-shot TetraBFT with failed blocks.}
\label{fig:block_fail}
\end{figure}

We also provide an example illustrating a block failure leading to a view change (refer to Fig.~\ref{fig:block_fail}). Starting with view 0 of slot 1, a view timer $t_1$ starts. A well-behaved leader proposes a block with value $val_1$, and subsequently nodes vote for $val_1$ and start a view timer $t_2$ for slot 2. Concurrently, a well-behaved leader for slot 2 proposes a block with value $val_2$. Upon receiving a quorum of \texttt{vote} messages $\langle \texttt{vote}, slot-1, view-0, value-val_1\rangle$, a well-behaved node votes for $val_2$. 
However, if a well-behaved node fails to receive a quorum of \texttt{vote} messages for $val_2$ by the time $t_1$ expires, it results in aborting the block for slot 1, consequently, the block for slot 2. In response, the node broadcasts a \texttt{view-change} message $\langle \texttt{view-change}, slot-1, view-1\rangle$. When a node receives $n-f$ \texttt{view-change} messages for view 1, it changes to view 1 and resets the view timer $t_1$. It then sends two \texttt{proof} messages $\langle \texttt{vote-1}, slot-1, view-0, value-val_1\rangle$, $\langle \texttt{vote-1}, slot-2, view-0, value-val_2\rangle$ and one \texttt{suggest} message $\langle \texttt{vote-2}, slot-1, view-0, value-val_1\rangle$ (notably, although \texttt{vote} messages do not explicitly indicate the phase, this information is preserved in the local memory). 
Upon receiving enough \texttt{suggest} messages, a well-behaved leader for view 1 and slot 1 proposes a new block with a safe value $val_1'$, adhering to \Cref{rule:picking-safe-proposal}. On receiving the \texttt{proposal} message with $val_1'$, a well-behaved node resets the view timer $t_2$ and advances to view 1 of slot 2. On determining $val_1'$ safe according to \Cref{rule:checking-safe-proposal}, a well-behaved node broadcasts a \texttt{vote} message $\langle \texttt{vote}, slot-1, view-1, value-val_1'\rangle$. Simultaneously, a well-behaved leader for view 1 and slot 2 proposes a block with a safe value $val_2'$, following \Cref{rule:picking-safe-proposal}. After determining $val_2'$ safe and receiving a quorum of \texttt{vote} messages for $val_1'$, a well-behaved node broadcasts a \texttt{vote} message $\langle \texttt{vote}, slot-2, view-1, value-val_2'\rangle$. The process for slot 3 mirrors that for slot 2, and is thus not detailed further. As no block for slot 4 was proposed previously, it defaults to starting from view 0. Slot 4 and subsequent slots proceed as in the good case described earlier. If a quorum of \texttt{vote} messages for slot 4 is received before the expiration of timer $t_1$, the block for slot 4 is notarized and the block for slot 1 is finalized.

\begin{algorithm}
\caption{$view\_change(v,s)$}\label{alg:view change}
\begin{algorithmic}[1]
\NoNumber{Code for node $i$:}
\State $proposed\_v\gets v+1$
\While{$true$}
    \Upon{receiving $f+1$ $\langle \texttt{view-change}, slot-s', view-v'\rangle$ messages}
    \If{$v'>proposed\_v$}
        \State $proposed\_v\gets v'$
        \State broadcast a $\langle \texttt{view-change}, slot-s', view-proposed\_v\rangle$ message
    \EndIf
    \EndUpon
    \Upon{receiving $n-f$ $\langle \texttt{view-change}, slot-s'', view-v''\rangle$ messages}
        \State $v\gets v''$
        \State $s\gets s''$
        \State abort blocks for $s, s+1,\ldots$ (view timers $t_s, t_{s+1},\ldots$ are invalid)
        \State broadcast a \texttt{suggest} and a \texttt{proof} message for block $s, s+1,\ldots$
    \EndUpon
\EndWhile
\end{algorithmic}
\end{algorithm}

\begin{algorithm}
\caption{Multi-shot TetraBFT}\label{alg:chained protocol}
\begin{algorithmic}[1]
\NoNumber{Code for node $i$:}
\State $v\gets 0, s\gets 1$
\While {$true$}
    \State continually \textbf{run} $view\_change(v,s)$ in the background
    \State a view timer $t_s$ starts \textcolor{gray}{// a new slot $s$ begins}
    \State continually \textbf{run} lines 6-8 in the background
    \For{$t_{s_k}\in (t_{s-3},t_{s-2},t_{s-1})$}
        \AtTime{$t_{s_k}+9\Delta$}
            \State broadcast a $\langle \texttt{view-change}, slot-s_k, view-v_{s_k}+1\rangle$ message
        \EndAtTime
    \EndFor
    \If{block $s$ is not proposed before}
        \State $v\gets 0$
    \EndIf
    \If{$s=1$ \textbf{and} node $i$ = leader}
        \State propose block $s$ with a safe value $val$ according to \Cref{rule:picking-safe-proposal}
    \EndIf
    \If{node $i \neq$ leader \textbf{and} determines block $s$'s value $val$ safe according to \Cref{rule:checking-safe-proposal}}
        \State broadcast a $\langle \texttt{vote}, slot-s, view-v, value-val\rangle$ message 
    \EndIf
    \If{node $i$ = leader}
        \State propose block $s+1$ with a safe value $val'$ according to \Cref{rule:picking-safe-proposal}
    \EndIf
    \Upon{receiving $n-f$ $\langle \texttt{vote}, slot-s, view-v, value-val\rangle$ messages}
        \State notarize block $s$ with value $val$
        \If{blocks $s-3,s-2,s-1,s$ are notarized}
            \State finalize block $s-3$
        \EndIf
        \State $s\gets s+1$
    \EndUpon
\EndWhile
\end{algorithmic}
\end{algorithm}

\subsection{Security Analysis for Multi-Shot TetraBFT}
\label{Security Analysis for Multi-shot TetraBFT}

\begin{theorem}
Multi-shot TetraBFT guarantees consistency and liveness.
\end{theorem}
\begin{proof}
In the Multi-shot TetraBFT protocol, every vote serves multiple purposes: it acts as a proposal for the current block $s$, and simultaneously as \texttt{vote-1} for block $s-1$, \texttt{vote-2} for block $s-2$, \texttt{vote-3} for block $s-3$, and \texttt{vote-4} for block $s-4$. 
Each vote is sent only after receiving a quorum from the preceding vote. 
In a view-change case, if a node receives $n-f$ \texttt{view-change} messages for a view and slot, it changes to the new view for that block. This scenario necessitates the sending of \texttt{suggest}/\texttt{proof} messages as well. In the Multi-shot TetraBFT protocol, the structure of voting and the process for view change are consistent with basic TetraBFT, allowing us to apply the properties of basic TetraBFT to validate Multi-shot TetraBFT.

\textbf{Consistency.} Assuming two chains are finalized by two well-behaved nodes, and if one chain is neither a prefix of the other nor identical, this situation leads to the finalization of two distinct blocks within the same slot. However, according to the agreement property of the basic TetraBFT, no two well-behaved nodes can decide different values. This scenario presents a contradiction to the agreement property of basic TetraBFT.

\textbf{Liveness.} We consider after GST that all messages are guaranteed to be delivered. Consider a scenario where a well-behaved node receives a sequence of proposals (transactions) from well-behaved leaders. Consequently, all well-behaved nodes will receive these proposals from the leaders. Leveraging the validity and termination properties of the basic TetraBFT protocol, each well-behaved node will make a decision on each proposal. In the good case, every proposal is decided and notarized as a block. In the view-change case, after $5\Delta$ time ($2\Delta$ for view change and $3\Delta$ for \texttt{suggest/proof} message, \texttt{proposal} message, and a \texttt{vote} message), a new block gets notarized.  Given that every sequence of four consecutive blocks gets notarized, the first block is finalized. Therefore, as long as there are blocks that have been notarized, it follows that the transactions received can also be finalized.
\end{proof}

\section{Conclusion and Future Work}

In this paper, we have improved upon the state-of-the-art in optimally-resilient, optimistically responsive, partially-synchronous unauthenticated BFT consensus protocols, deploying new algorithmic principles to propose TetraBFT, which achieves better latency than IT-HS~\cite{abraham2020information} while guaranteeing constant storage and linear communicated bits per node per view.

Moreover, we have proposed a pipelined version of TetraBFT which, at least in theory, multiplies the throughput of TetraBFT by a factor of 5, and we have conducted the first detailed analysis of a pipelined protocol in the unauthenticated setting.

Taking stock, we asked in the introduction what is the minimum good-case latency achievable for optimistic responsiveness with constant space.
TetraBFT shows that the latency of 6 message delays achieved by IT-HS is not optimal, as it achieves 5.
However, whether 5 is the lower bound remains an open question.
Anecdotal evidence could point towards a lower bound of 5, as a recent investigation of the connected-consensus problem~\cite{attiya_multi-valued_2023} (a generalization of adopt-commit and graded agreement), which is solvable asynchronously, also arrived at an unauthenticated BFT protocol with 5 phases.

In future work, it would also be interesting to implement Multi-shot Tetra BFT and conduct a practical evaluation, in particular in the context of heterogeneous trust systems such as the XRP Ledger or the Stellar network, where TetraBFT has the potential to offer better guarantees than currently deployed protocols, and in the context of trustless cross-chain synchronization protocols executed by smart contracts.

\begin{acks}
This work is supported in part by a gift from Stellar Development Foundation and by the Guangzhou-HKUST(GZ) Joint Funding Program (No. 2024A03J0630). 
\end{acks}

\bibliographystyle{ACM-Reference-Format}
\bibliography{8_bib.bib}
\clearpage
\appendix
\section*{Appendix}
\section{Helper Algorithms}
\label{sec:helper algo}

\begin{algorithm}[H]
\caption{A leader determines value $val$ safe in view $v$}\label{alg:leader finds v}
\begin{algorithmic}[1]
    \NoNumber{Input: \texttt{suggest} messages($\{suggest[vote\mathit{2}[view, val], prev\_vote\mathit{2}[view, val]], vote\mathit{3}[view, val]\}$), $v$}
    \State $vote\mathit{3}_\_val\_set \gets \emptyset,vote\mathit{2}_\_val\_set \gets \emptyset$
    \State $val\gets val_{init},count\_no\_vote\mathit{3} \gets 0$
    \State $vote\mathit{2}\_view\_val\_dict\gets \{\}, prev\_vote\mathit{2}\_view\_val\_dict\gets \{\}$
    \Statex \textcolor{gray}{// $vote\mathit{2}\_view\_val\_dict: \{view: (cnt,vote\mathit{2}\_inner\_val\_set)\}$}
    \Statex \hspace{0.4cm} \textcolor{gray}{$prev\_vote\mathit{2}\_view\_val\_dict:\{view, cnt\}$}
    \ForAll{$suggest\in$ \texttt{suggest} messages}
        \If{$suggest.vote\mathit{3}= \emptyset$}
            \Comment{\Cref{rule:picking-safe-proposal} \Cref{case:no_vote_leader}}
            \State $count\_no\_vote\mathit{3} \gets count\_no\_vote\mathit{3} + 1$ 
            \If{$count\_no\_vote\mathit{3} \ge n-f$}
                \State \Return $val$
            \EndIf
        \ElsIf{$suggest.vote\mathit{3}.val\notin vote\mathit{3}\_val\_set$}
            \State add $suggest.vote\mathit{3}.val$ to $vote\mathit{3}\_val\_set$
        \EndIf
        \If{$suggest.vote\mathit{2}\neq \emptyset$ and $suggest.prev\_vote\mathit{2}= \emptyset$}
            \State $vote\mathit{2}\_view\_val\_dict[suggest.vote\mathit{2}.view]\gets\{cnt\gets cnt+1$ if $suggest.vote\mathit{2}.view$ exists 
            \Statex \hspace{1cm} in the dict else 1, add $suggest.vote\mathit{2}.val$ to $vote\mathit{2}\_inner\_val\_set\}$
        \EndIf
        \If{$suggest.prev\_vote\mathit{2}\neq \emptyset$}
            \State $prev\_vote\mathit{2}\_view\_val\_dict[suggest.prev\_vote\mathit{2}.view]\gets\{cnt\gets cnt+1$ 
            \Statex \hspace{1cm} if $suggest.prev\_vote\mathit{2}.view$ exists in the dict else 1$\}$
        \EndIf
    \EndFor
    \For{$v' \gets v-1, v-2, \ldots,0$} 
        \Comment{\Cref{rule:picking-safe-proposal} \Cref{case:highest_vote_leader}}
        \State $quorum\_num \gets 0,blocking\_num \gets 0$
        \State $vote\mathit{2}\_view\_ge\_v'\_cnt\gets \sum vote\mathit{2}\_view\_val\_dict[view][cnt]$ if $view\ge v'$ 
        \State $prev\_vote\mathit{2}\_view\_ge\_v'\_cnt\gets \sum prev\_vote\mathit{2}\_view\_val\_dict[view][cnt]$ if $view\ge v'$ 
        \If{$vote\mathit{2}\_view\_ge\_v'\_cnt + prev\_vote\mathit{2}\_view\_ge\_v'\_cnt< f+1$} \label{algo:2_19}
            \State \textbf{continue}
        \Else
            \State $vote\mathit{2}_\_val\_set \gets vote\mathit{2}_\_val\_set\cup vote\mathit{2}\_view\_val\_dict[v'][vote\mathit{2}\_inner\_val\_set]$
            \ForAll{$val \in  vote\mathit{2}_\_val\_set\cup vote\mathit{3}_\_val\_set$}
                \ForAll{$suggest\in$ \texttt{suggest} messages}
                    \If{$suggest.vote\mathit{3}.view < v'$ \textbf{or} ($suggest.vote\mathit{3}.view=v'$ \textbf{and} 
                    \Statex \hspace{3cm} $suggest.vote\mathit{3}.val=val$)}
                        \State $quorum\_num \gets quorum\_num + 1$
                        \If{$node\_claim\_safe(suggest,v',val)=true$}
                            \State $blocking\_num \gets blocking\_num + 1$
                        \EndIf
                    \EndIf
                    \If{$quorum\_num \ge n-f$ \textbf{and} $blocking\_num \ge f+1$}
                        \State \Return $val$
                    \EndIf
                \EndFor
            \EndFor
        \EndIf
    \EndFor
\State \Return $val$
\end{algorithmic}
\end{algorithm}

\begin{algorithm}
\caption{A node determines value $val$ safe in view $v$}\label{alg:node determine value}
\begin{algorithmic}[1]
\NoNumber{Input: \texttt{proof} messages($\{proof[vote\mathit{1}[view, val], prev\_vote\mathit{1}[view, val]], vote\mathit{4}[view, val]\}$), $v$, $val$}
\State $count\_no\_vote\mathit{4} \gets 0$
\State $val\_view\_dict=\{\}$ \textcolor{gray}{// $val\_view\_dict=\{(view,val):proof\_message\_inner\_set$\}}
\State $vote\mathit{1}\_val\_set\gets \emptyset$
\ForAll{$proof\in$ \texttt{proof} messages} 
    \If{$proof.vote\mathit{4}= \emptyset$} \Comment{\Cref{rule:checking-safe-proposal} \Cref{case:no_votes}}
        \State $count\_no\_vote\mathit{4} \gets count\_no\_vote\mathit{4} + 1$
        \If{$count\_no\_vote\mathit{4} \ge n-f$} 
            \State \Return $true$
        \EndIf
    \EndIf
    \If{$proof.vote\mathit{1}\neq \emptyset$ and $proof.prev\_vote\mathit{1}= \emptyset$}
        \State add $proof.vote\mathit{1}.val$ to $vote\mathit{1}\_val\_set$
    \EndIf
\EndFor
\For{$v'\gets v-1, \ldots,0$} \Comment{\Cref{rule:checking-safe-proposal} \Cref{case:one block set}} \label{algo:3_11}
    \State $quorum\_num \gets 0,blocking\_num \gets 0$
    \ForAll{$proof\in$ \texttt{proof} messages}
        \If{$proof.vote\mathit{4}.view < v'$ \textbf{or} ($proof.vote\mathit{4}.view=v'$ \textbf{and} $proof.vote\mathit{4}.val=val$)}
            \State $quorum\_num \gets quorum\_num + 1$
            \If{$node\_claim\_safe(proof,v',val)=true$}
                \State $blocking\_num \gets blocking\_num + 1$
            \EndIf
        \EndIf
        \If{$quorum\_num \ge n-f$ \textbf{and} $blocking\_num \ge f+1$}
            \State \Return $true$
        \EndIf
    \EndFor
\EndFor \label{algo:3_19}
\For{$view\gets v-1, \ldots,0$} \Comment{\Cref{rule:checking-safe-proposal} \Cref{case:two block sets}}
    \ForAll{$val \in vote\mathit{1}\_val\_set$}
        \State $node\_claim\_safe\_cnt\gets 0, proof\_inner\_set \gets \emptyset$
        \ForAll{$proof\in$ \texttt{proof} messages}
            \If{$node\_claim\_safe(proof,view,val) = true$}
                \State $node\_claim\_safe\_cnt\gets node\_claim\_safe\_cnt + 1$
                \State add $val$ to $proof\_inner\_set$
            \EndIf
        \EndFor
        \If{$node\_claim\_safe\_cnt\ge f+1$}
            \State $val\_view\_dict[(view,val)]\gets \{proof\_inner\_set$\}
            \ForAll{$\{(view', val'): proof\_inner\_set'\}$ in $val\_view\_dict$
                \Statex \hspace{2.3cm} if ($view'>view$ \textbf{and} $val'\neq val$)}
                \State $quorum\_num \gets 0,proof\_quorum\_set\gets \emptyset$
                \ForAll{$proof\in$ \texttt{proof} messages}
                    \Statex \hspace{2.3cm} \textcolor{gray}{// it is sufficient to check if \Cref{rule:checking-safe-proposal} \Cref{2bi,2bii} hold only when $v'=view$}
                    \If{$proof.vote\mathit{4}.view < view$ 
                    \Statex \hspace{3cm} \textbf{or} ($proof.vote\mathit{4}.view=view$ \textbf{and} $proof.vote\mathit{4}.val=val$)}
                        \State $quorum\_num\gets quorum\_num+1$
                        \State add $proof$ to $proof\_quorum\_set$
                    \EndIf
                \If{$quorum\_num \ge n-f$ \textbf{and}
                    \Statex \hspace{3cm} $|proof\_inner\_set\cap proof\_quorum\_set|\ge f+1$ \textbf{and}
                    \Statex \hspace{3cm} $|proof\_inner\_set'\cap proof\_quorum\_set|\ge f+1$}
                    \State \Return $true$
                \EndIf
                \EndFor
            \EndFor
        \EndIf
    \EndFor
\EndFor
\State \Return $false$
\end{algorithmic}
\end{algorithm}

\clearpage
\onecolumn
\section{Formal Specification of TetraBFT in TLA+}
\label{sec:formal-spec}
\includegraphics[page=1, width=\textwidth]{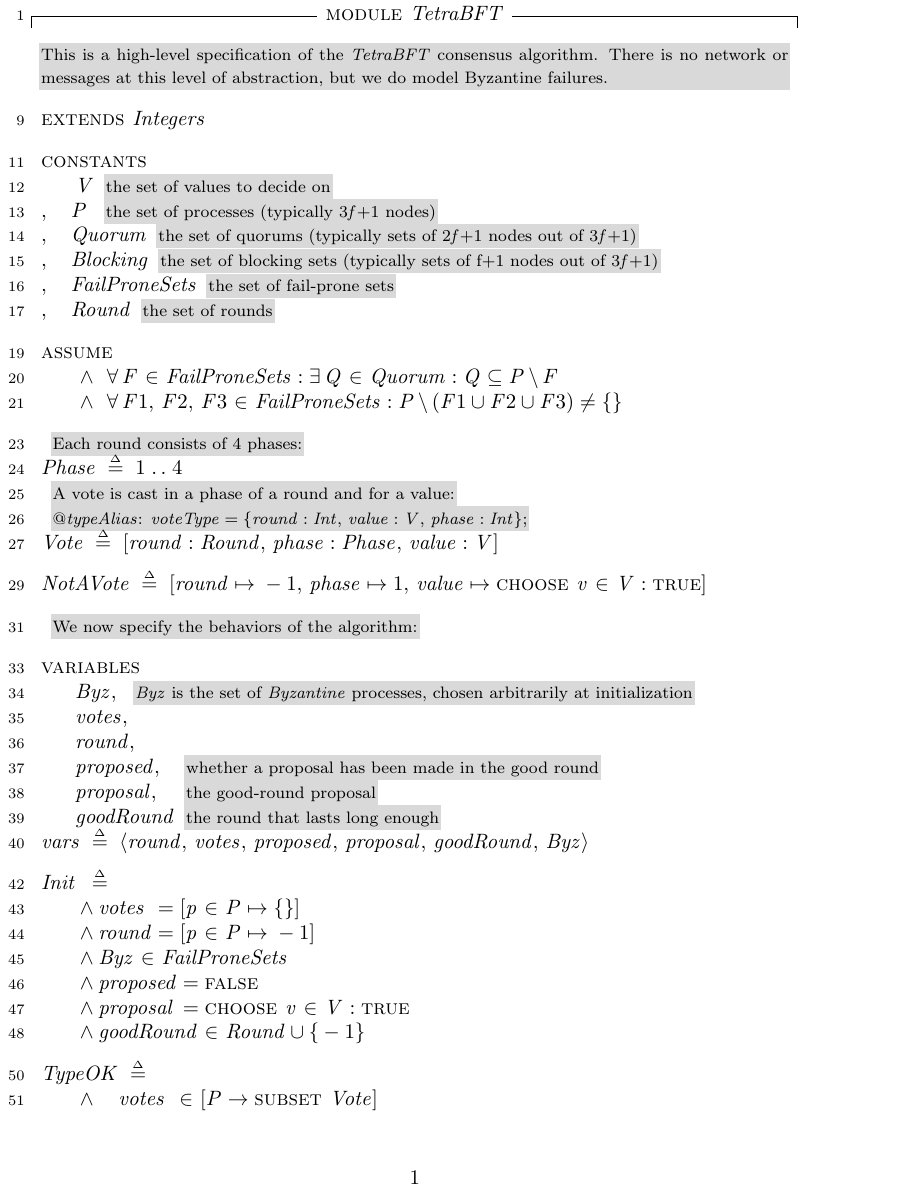}
\newpage
\includegraphics[page=2, width=\textwidth]{TetraBFT.pdf}
\newpage
\includegraphics[page=3, width=\textwidth]{TetraBFT.pdf}
\newpage
\includegraphics[page=4, width=\textwidth]{TetraBFT.pdf}
\newpage
\includegraphics[page=5, width=\textwidth]{TetraBFT.pdf}
\newpage
\includegraphics[page=6, width=\textwidth]{TetraBFT.pdf}
\newpage
\includegraphics[page=7, width=\textwidth]{TetraBFT.pdf}

\end{document}